\documentclass[draftcls,onecolumn,draftclsnofoot,twoside,letter,12 pt ]{IEEEtran}
\usepackage{graphicx}
\usepackage{caption}
\usepackage{subcaption}
\usepackage{amsmath}

\newcounter{subeqn} %
\makeatletter
\@addtoreset{subeqn}{equation}
\makeatother

\usepackage[ruled,vlined]{algorithm2e}
\usepackage{amssymb}
\usepackage{mathrsfs}
\usepackage{stfloats}
\usepackage{dsfont}
\usepackage{amsthm}
\usepackage{tikz}
\usepackage{epstopdf}
\usepackage{epsfig}
\usepackage{pstricks}
\usepackage{hyperref}
\usepackage[noadjust]{cite}
\usepackage{filecontents}

\newtheorem{proposition}{Proposition}

\newtheorem{remark}{Remark}

\begin{document}
\title{Joint Constellation Rotation and Symbol-level Precoding Optimization in the Downlink of Multiuser MISO Channels}

%
\author{
  \IEEEauthorblockN{ Maha~Alodeh,~\IEEEmembership{Member, IEEE}, Björn Ottersten~\IEEEmembership{Fellow Member, IEEE} }\\
\thanks{Maha Alodeh is with Ericsson AB, Address: Torshamnsgatan 21, 164 40 Stockholm, e-mails:\{ maha.mahmoud.hamed.alodeh@ericsson.com, maha.alodeh@gmail.com\}. Björn Ottersten is with Interdisciplinary Centre for Security, Reliability and Trust (SnT), University of Luxembourg,  email:\{bjorn.ottersten@uni.lu\}. The proposed ideas and algorithm of this paper is filed under Downlink Symbol-Level Precoding 
	As Part Of Multiuser MISO Communication, P80999WO1.
	 Maha Alodeh is the corresponding author.   
	
	}  
}
 

\providecommand{\keywords}[1]{\textbf{\textit{Index terms---}} #1}

\date{}

\maketitle

\begin{abstract}
This paper tackles the problem of the simultaneous interference among the multiple users in the downlink of a wireless multiantenna system. In order to exploit the multiuser interference and transform it into useful power at the receiver side, symbol-level precoding is utilized through  joint exploitation of the data from each individual user and the channel state information. This paper introduces a complexity reduction scheme that exploits the constellations symmetry to reduce the number of possible calculations and the number of constraints in the symbol-level optimization problem.
This paper also proposes a new transmission technique that jointly optimizes the transmit symbol-level precoding  and the constellation rotation of the data stream for each user. The purpose is to increase the probability of having constructive interference, and thus better performance. We focus on improving the energy efficiency by reducing the required power to satisfy certain quality of service constraints through solving a symbol-level precoding problem and finding the optimal phase by which each user constellation should be rotated. This problem is non-convex,  therefore, we propose an efficient branch-and-bound algorithm combined with semidefinite relaxation
to solve the problem and find an efficient solution. Numerical results are presented in a comparative fashion to show the effectiveness of the proposed technique, which outperforms the state-of-the-art symbol-level precoding schemes in terms of energy efficiency, and power consumption.

\end{abstract}
\section{Introduction}
\IEEEPARstart{I}{nterference}
is one of the crucial and limiting factors in wireless networks. Traditionally, time and frequency resources are are allocated and utilized to allow different users to communicate without inducing harmful interference. The
concept of exploiting the users' spatial separation for multiple access has been a fertile research domain for more than three
decades. This can be implemented by adding multiple antennas at one or both
communication sides. Multiantenna transceivers empower communication systems with more degrees
of freedom boosting the performance if the multiuser interference is mitigated
properly. Exploiting the spatial dimension, to serve different users simultaneously
in the same time slot and the same frequency band through spatial division multiple access (SDMA), has been  investigated thoroughly in literature \cite{bjorn_patent,mats,bjornson,Swindlehurst_zeroforcing,leakage_per,leakage,Medra_transaction,Medra}.

Precoding has played
a key role in the downlink of multiuser multiple-input-single-output (MISO) systems enhancing throughput and energy efficiency. Multiuser interference is mitigated by transmitting superimposed signals that are pre-designed at the base station. However, to produce efficient precoding techniques, the base station should acquire or have access to accurate channel state information (CSI). The precoding strategies need to be carefully designed to mitigate interference while optimizing a certain objective performance\cite{mats,leakage,Medra,Medra_transaction}.

 Precoding design can be classified based on two criteria: switching rate and number of users per stream (group size) \cite{maha_survey}.  The switching rate is the rate at which the precoding vectors change, it addresses how many precoding vectors are required to be designed in the coherence assuming the same users set. In this category, there are two precoding strategies has been proposed in the literature: symbol-level precoding \cite{Alireza_LSP1,Alireza_LSP2,Alireza_TSP,Alodeh_ICASSP2017,Jevgevij,Ma,Masouros_made,Spano_ICASSP2018,ashkan,directional modulation1,directional modulation2,doumachtis,maha_TSP,maha_icc_2018,maha_twc_1,maha_twc_2,masouros_1,masouros_2,sohrabi,swindlehurst_slp} and block-level precoding \cite{Chinese_norm,Medra,Medra_transaction,Swindlehurst_zeroforcing,bjorn_patent,leakage,leakage_per,mats}. In block-level precoding, the precoding vectors are designed depending only on CSI and the power of the users, they are fixed as long as CSI and/or the set of served users do not change. On the other hand, the symbol-level precoding, the precoding vector calculations depend on the data symbols and CSI. Therefore, the precoding vector changes with the rate of the data symbols and CSI assuming the same set of users. The block-level precoding designs its vector based on CSI to reduce the leakage of the interference \cite{leakage,leakage_per,Swindlehurst_zeroforcing} or to maximize the energy efficiency \cite{mats}. However, their main purpose is to mitigate interference and treat it only as a harmful factor that can hinder the performance of the wireless communication systems.

Symbol-level precoding (SLP) was proposed to exploit the interference in the downlink of multiuser MISO systems \cite{masouros_1,masouros_2,maha_TSP,maha_twc_1,maha_twc_2,spano_tsp,spano_twc,Spano_ICASSP2018,maha_icc_2018,maha_survey,ashkan,doumachtis,Alireza_TSP,Alireza_LSP1,Alireza_LSP2,Jevgevij,Masouros_made,Ma,Alodeh_ICASSP2017,sohrabi}. This can be implemented by designing the precoding or output vectors symbol-by-symbol, which leads to a better exploiting of the interference nature and structure. The basic idea of symbol rate switching was initially proposed in the context of directional modulation in \cite{directional modulation1,directional modulation2} to improve the security of the wireless systems without discussing its capability of exploiting the interference. Most of the symbol-level precoding literature tackles the problem of energy-efficiency in single-level modulations ($M$-PSK) \cite{masouros_1,masouros_2,maha_TSP,maha_twc_1}. In \cite{maha_twc_2,spano_tsp,Alireza_TSP,Alireza_LSP1,Alireza_LSP2}, the proposed precoding schemes are generalized to any generic modulation. A per-antenna consideration is thoroughly discussed in \cite{spano_tsp}, where  strategies based on the minimization of the power peaks
amongst the transmitting antennas and the reduction of the
instantaneous power imbalances across the different transmitted
streams are investigated. In \cite{doumachtis}, the considered systems tackle the high hardware complexity and power consumption of existing SLP techniques by reducing or completely eliminating fully digital radio frequency (RF) chains. The connection of SLP with the conventional zero-forcing (ZF) precoding is discussed in \cite{Masouros_made} for single-level modulation. It also proposes an iterative closed-form scheme to obtain the optimal beamforming matrix, where within each iteration a closed-form solution can be obtained.  An SLP approach for generic constellations with any arbitrary shape and size is considered in \cite{Alireza_TSP,Alireza_LSP1,Alireza_LSP2}. The proposed algorithm results in less computational complexity precoding schemes for different optimization problems.

Conventionally in SLP, the output vector is calculated per input vector, this leads to huge computational complexity especially at high modulations order and number of users. Also, the output vector is designed to guarantee that each symbol should be received in the correct detection region without any constellation rotation. However, optimized rotation for each user constellation can improve the performance of SLP by aligning the symbols to increase the probability of having constructive interference, which can help pushing them deeper in their correct detection region. This results in better symbols alignments that can maximize the energy efficiency. The optimal phase rotations depend on the multiuser channel of the users that are served simultaneously. These two aspect have not been addressed in literature before and they are the main focus of this work.

Our contribution is twofold. First, the paper introduces a complexity reduction approach by either reducing the number of constraints to be considered in single optimization problem, or the number of optimization problems to be solved within channel coherence time. This is done by exploiting the symmetrical nature of the constellations, which allows to calculate the output vectors for subset of data vectors ; this can reduces the calculations to $1/4$ of the original SLP calculations without losing any performance gain. Second, it proposes a tailored efficient algorithm for optimizing the output vector and the optimal phase by which each user constellation should be rotated to guarantee the maximum energy efficiency, which contrasts to all existing works that focus on the design of the optimal symbol-level precoding without any constellation rotation. The challenging aspect in the joint optimization of the precoding and constellation rotation is the non-convex nature of the constraints in the optimization problem, which are constant modulus and bilinear. Our proposed algorithm is based on
the branch-and-bound strategy combined with an argument
cut technique \cite{Chinese_norm} and semidefinite relaxation \cite{convex_boyd} to solve the problem. The argument cuts are used to
design effective convex relaxations of non-convex constant modulus constraints in problem and therefore
play an important role in our proposed branch-and-bound
algorithm for solving the problem.
Since the joint constellation rotation and precoding design problem is
NP-hard, there does not exist a polynomial time algorithm which can solve it to global optimality.

The organization of this paper is as follows: In Section \ref{System model}, the system model is explained in details. In Section \ref{constructive interference},
we revise the basic concept of symbol-level precoding. A complexity reduction approach is proposed in \ref{complexity_red}. The problem of joint precoding and constellation rotation is illustrated in section \ref{Joint SLP}. An iterative algorithm based on semidefinite relaxation and argument cuts is thoroughly explained in section \ref{Algorithm}. Finally, simulation results are presented in section \ref{Results}
to illustrate the efficiency of our solution to the joint constellation rotation and precoding optimization problem.

Notation: we use boldface upper and lower case letters for
matrices and column vectors, respectively. $(\cdot)^H$, $(\cdot)^*$
 stand for Hermitian transpose and conjugate of $(\cdot)$. $\mathbb{E}(\cdot)$ and $\|\cdot\|$ denote the statistical expectation and the Euclidean norm. $\angle(\cdot)$, $|\cdot|$ are the angle and magnitude  of $(\cdot)$ respectively. $\mathcal{R}\{a\}$ and $\mathcal{I}\{a\}$ are the real part  and the imaginary part of $a$. $\mathbf{I}$ is identity matrix.  The curled inequality symbol $\succcurlyeq$  (its strict form
and reverse form ) is used to denote generalized inequality: $\mathbf{A}\succcurlyeq\mathbf{B}$ 
means that $\mathbf{A}-\mathbf{B}$ is an Hermitian positive semidefinite
matrix.
\section{System Model}
\label{System model}
Let us consider a single-cell multiple-antenna downlink system,
where a base-station equipped with $M$ transmit antennas delivers $K$ independent data streams to $K$
single-antenna user terminals, where
$K\leq M$. Each data stream is divided in blocks of $S$ symbols. Considering a data block, we can define the data information matrix
$\mathbf{D} \in \mathbb{C}^{K\times S}= [\mathbf{d}[1] \dots \mathbf{d}[S]]$, which aggregates the symbol streams
to be delivered to the different users, where $\mathbf{d}[n]=[d_1[n], \hdots, d_K[n]]^T$. We assume a quasi static block fading channel $\mathbf{h}_j\in\mathbb{C}^{1\times
	M}$ between
the BS antennas and the $j^{th}$ user, where the received signal at
j$^{th}$ user is
written as
\begin{eqnarray}
y_j[n]&=&\mathbf{h}_j\mathbf{x}[n]+z_j[n].
\end{eqnarray} $\mathbf{x}[n]\in\mathbb{C}^{M\times 1}$ is the transmitted signal vector from the multiple antennas
transmitter and  $z_j$ denotes the noise at receiver $j$, which is assumed to be independent and identically distributed complex Gaussian $\mathcal{CN}(0,1)$. A compact formulation
of the received signal at all users' receivers can be written as
\begin{eqnarray}
\mathbf{y}[n]&=&\mathbf{H}\mathbf{x}[n]+\mathbf{z}[n].
\end{eqnarray}
Let $\mathbf{x}[n]$ be written as $\mathbf{x}[n]=\sum^K_{j=1}\mathbf{w}_j[n]d_j[n]$,
where $\mathbf{w}_j$ is the $\mathbb{C}^{M\times
	1}$ unit power precoding vector for the user $j$. The received signal at $j^{th}$
user ${y}_j$ in $n^{th}$ symbol period is given by
\begin{eqnarray}
\label{rx_o}
{y}_j[n]=\mathbf{h}_j\mathbf{w}_j[n] d_j[n]+\displaystyle\sum_{k\neq j}\mathbf{h}_j\mathbf{w}_k[n]
d_k[n]+z_j[n].
\end{eqnarray}
The channel state information is assumed to be fully available at the base station. Without loss of generality, we assume that $\mathcal{S}=\mathcal{N}$ and $K=M$, $\mathbf{D} \in \mathcal{C}^{M\times\mathcal{N}}$ presents the brute force of all possible combination of input data vectors, which means $\mathbf{d}[n]\neq \mathbf{d}[m], \forall n, m \in \mathcal{N}$.

\section{Constructive Interference and Symbol-level Precoding}
\label{constructive interference}
In this section, we revise the definition of the constructive interference and the fundamental concept of symbol-level precoding technique\cite{maha_TSP,masouros_2}.   
\subsection{Constructive Interference}
Interference in the downlink of multiuser MISO system can be classified into constructive and destructive. The constructive interference manages to push the signal deeper in the correct detection region at the receiver, otherwise it is considered to be destructive.
A detailed definition for constructive interference is fully explained in \cite{maha_TSP}. In this paper, we describe a summary of the definition for the sake of manuscript completion. 

The constructive interference can be beneficial for single-level modulations such BPSK and $M$-PSK, or for outermost constellation points in multi-level modulations such $M$-QAM and APSK. The number of outermost constellation points that can benefit from the constructive interference is $\frac{M}{2}-1$.
  
An $M$-PSK modulated symbol $d_j$, is said to receive constructive
interference from another simultaneously transmitted symbol $d_k$ which is
associated with $\mathbf{w}_k$ if and only if the following inequalities hold   
\begin{equation}\nonumber
\label{one}
\angle{d_j}-\frac{\pi}{M}\leq \arctan\Bigg(\frac{\mathcal{I}\{\mathbf{h}_j\mathbf{w}_kd_{k}\}}{\mathcal{R}\{\mathbf{h}_j\mathbf{w}_kd_{k}\}}\Bigg)\leq \angle{d_j}+\frac{\pi}{M},
\end{equation}
\begin{equation}\nonumber
\label{two}
\mathcal{R}\{{d_j}\}.\mathcal{R}\{\mathbf{h}_j\mathbf{w}_kd_k\}>0, \mathcal{I}\{{d_j}\}.\mathcal{I}\{\mathbf{h}_j\mathbf{w}_kd_k\}>0,\\
\end{equation}
where $\psi_{jk}$ is the interference created from the trasmission to user $k$ from user $j$. 
\subsection{Symbol-level Precoding in the Downlink of Multiuser System}
\label{SLP Literature}
Interference is a random deviation that can move the desired constellation point in any direction. To address this problem, the power of interference has been used in the past to regulate its effect on the desired signal point. 
The interference among the multiuser spatial streams
leads to a deviation of the received symbols outside of their detection region. To tackle the interference efficiently, symbol-level precoding designs the output vector on symbol basis, which can be a powerful processing tool to exploit the multiuser interference, or to mitigate the non-linear effect. The optimization that
 minimizes the transmit power and 
 the constructive reception of the transmitted data symbols in the SLP context can be written
 as \cite{maha_TSP}
\begin{eqnarray}\nonumber
&\hspace{-0.5cm}\mathbf{w}_k& =\arg\underset{\mathbf{w}_k}{\min}\quad \|\sum^K_{k=1}\mathbf{w}_kd_k\|^2\\
&\hspace{-0.5cm}s.t.& \begin{cases}\mathcal{C}_1:\mathcal{R}\{\mathbf{h}_j\sum^K_{k=1}\mathbf{w}_kd_k\}-\sigma_z\sqrt{\gamma_j}\mathcal{R}\{{d}_j\}\unlhd {0}, \forall j\in \mathcal{K},\\
\mathcal{C}_2:\mathcal{I}\{\mathbf{h}_j\sum^K_{k=1}\mathbf{w}_kd_k\}-\sigma_z\sqrt{\gamma_j}\mathcal{I}\{{d}_j\}\unlhd {0},\forall j\in \mathcal{K}, 
\end{cases}
\end{eqnarray}
using $\mathbf{x}=\sum^K_{k=1}\mathbf{w}_kd_k$, the problem can be reformulated as:
\begin{eqnarray}\nonumber
\label{SLP_conventional}
&\hspace{-0.5cm}\mathbf{x}_{opt}& =\arg\underset{\mathbf{x}}{\min}\quad \|\mathbf{x}\|^2\\
&\hspace{-0.5cm}s.t.& \begin{cases}\mathcal{C}_1:\mathcal{R}\{\mathbf{h}_j\mathbf{x}\}-\sigma_z\sqrt{\gamma_j}\mathcal{R}\{{d}_j\}\unlhd {0}, \forall j\in \mathcal{K}, \\
\mathcal{C}_2:\mathcal{I}\{\mathbf{h}_j\mathbf{x}\}-\sigma_z\sqrt{\gamma_j}\mathcal{I}\{{d}_j\}\unlhd {0},\forall j\in \mathcal{K}.\\
\end{cases}
\end{eqnarray}
$\gamma_j$ is the signal to noise ratio (SNR) target for user $j$. $d_j$ is the data symbol for the user $j$, belongs to to any selected  constellation. $N$ is the number of possible data vectors $2^{\sum^K_{j}m_j}$,  where $m_j$ is the modulation order assigned to user $j$.
$\unlhd$ is an element-wise operator to guarantee that each symbol is received in the correct detection region. This problem was solved efficiently in the literature using different techniques. The problem in \eqref{SLP_conventional} can be expanded to evaluate the output vector $\mathbf{x}$ for different input data vectors, possibly brute force all possible combinations, in one optimization problem as 
\begin{eqnarray}\nonumber
\label{SLP_conventional_temporal}
&\hspace{-0.5cm}\mathbf{x}_{opt}[n]& =\arg\underset{\mathbf{x}[n]}{\min}\quad \frac{1}{N}\sum^{N}_{n=1}\|\mathbf{x}[n]\|^2\\ 
&\hspace{-0.5cm}s.t.& \begin{cases}\mathcal{C}_1:\mathcal{R}\{\mathbf{h}_j\mathbf{x}[n]\}-\sigma_z\sqrt{\gamma_j}\mathcal{R}\{{d}_j[n]\}\unlhd {0},\\ \forall j\in \mathcal{K}, \forall n \in \mathcal{N}\\
\mathcal{C}_2:\mathcal{I}\{\mathbf{h}_j\mathbf{x}[n]\}-\sigma_z\sqrt{\gamma_j}\mathcal{I}\{{d}_j[n]\}\unlhd {0},\\ \forall j\in \mathcal{K}, \forall n \in \mathcal{N}.\\
\end{cases}.
\end{eqnarray}
 The difference between \eqref{SLP_conventional} and \eqref{SLP_conventional_temporal} is the introduction of the temporal dimension in \eqref{SLP_conventional_temporal}. This leads to solve a single quadratic optimization with $2\mathcal{K}\mathcal{N}$ affine constraints rather than $\mathcal{N}$ quadratic optimization problems with $2\mathcal{K}$ affine constraints.  However, initial results indicate that solving \eqref{SLP_conventional_temporal} requires less execution time than \eqref{SLP_conventional} for $\mathcal{N}$ times\cite{spano_twc}. If the number of possible input data vectors $\mathcal{N}$ is larger than than the precoded symbols in the frame, it is better to evaluate the output vectors for the input data vectors that are included in the frame. This may happen when there is very large number of users that are served simultaneously. On the other hand, solving \eqref{SLP_conventional}-\eqref{SLP_conventional_temporal} leads to same results, therefore, the system designer can select which optimization to solve. However, in the current SLP techniques \cite{Alireza_LSP1,Alireza_LSP2,Alireza_TSP,Jevgevij,Ma,Masouros_made,Spano_ICASSP2018,ashkan,doumachtis,maha_TSP,maha_icc_2018,maha_twc_1,maha_twc_2,masouros_2,masouros_1,spano_tsp,spano_twc,Alodeh_ICASSP2017,maha_survey}, the output vector is designed to direct the signal in the correct detection region, without any constellation rotation optimization. 

 \section{Complexity Reduction}
 \label{complexity_red}

 The bottleneck of employing SLP in current communication systems is the number of calculations, which is a function of channel state information and the number of possible data vectors $\mathbf{d}[n]$ per channel realization. The number of possible output vector calculations equals to $(2^{\sum_jm_j})$ for each channel realization in \eqref{SLP_conventional} or  $2^{1+\sum_jm_j}$ constraints in the single optimization problem interpretation in \eqref{SLP_conventional_temporal}, this leads to huge computational complexity especially for large number of users and/or high modulation orders.
 
  One of the constellation's distinguishing feature is its symmetry. This can be utilized to reduce the number of constraints in optimization problem \eqref{SLP_conventional_temporal}, or the number of optimization problems to be solved \eqref{SLP_conventional}. We can find some relations among different data vectors $\mathbf{d}[n]$ and how this reflects on the relation among their corresponding optimal output vectors $\mathbf{x}[n]$. For example, the data vectors $\mathbf{d}[n]$ and $\mathbf{d}[m]$ have the relation as $\mathbf{d}[m]=-\mathbf{d}[n]$,  the output vector relation should be $\mathbf{x}[m]=-\mathbf{x}[n]$ and thus the same amount of power is required to achieved the target SNR. Based on this observation, we can state a remark as
  
  \begin{remark}
When the data vectors exhibit the relation $\mathbf{d}[n]=\zeta\mathbf{d}[m]$, the output vector shall also have $\mathbf{x}[n]=\zeta\mathbf{x}[m]$  where $\zeta\in\{1, -1, -j, j\}$.   	
  \end{remark}
  
  \begin{proof}
  We assume that the optimal solution for data vector $\mathbf{d}[n]$ equals to $\mathbf{x}[n]$, then, the optimal output vector $\mathbf{x}[m]$ for the data vector $\mathbf{d}[m]=\zeta \mathbf{d}[n]$ can be found by solving the following optimization
  \begin{eqnarray}\nonumber
  \label{SLP_conventional}
  &\hspace{-0.5cm}\mathbf{x}_{opt}[m]& =\arg\underset{\mathbf{x}[m]}{\min}\quad \|\mathbf{x}[m]\|^2\\\nonumber
  &\hspace{-0.5cm}s.t.& \begin{cases}\mathcal{C}_1:\mathcal{R}\{\mathbf{h}_j\mathbf{x}[m]\}-\sigma_z\sqrt{\gamma_j}\mathcal{R}\{\zeta{d}_j[n]\}\unlhd {0}, \forall j\in \mathcal{K}, \\
  \mathcal{C}_2:\mathcal{I}\{\mathbf{h}_j\mathbf{x}[m]\}-\sigma_z\sqrt{\gamma_j}\mathcal{I}\{\zeta{d}_j[n]\}\unlhd {0},\forall j\in \mathcal{K}.\\
  \end{cases}
  \end{eqnarray}
  Substituting $\mathbf{d}[m]=\zeta\mathbf{d}[n]$, we can reformulate the optimization problem as
    \begin{eqnarray}\nonumber
  \label{SLP_conventional}
  &\hspace{-1.5cm}\mathbf{x}_{opt}[m]& =\arg\underset{\mathbf{x}[m]}{\min}\quad \|\mathbf{x}[m]\|^2\\\nonumber
  &\hspace{-1cm}s.t.& \begin{cases}\mathcal{C}_1:\mathcal{R}\{\mathbf{h}_j\mathbf{x}[m]\}\unlhd\\\sigma_z\sqrt{\gamma_j}(\mathcal{R}\{\zeta\}\mathcal{R}\{{d}_j[n]\} -\mathcal{I}\{\zeta\}\mathcal{I}\{{d}_j[n]\}) , \forall j\in \mathcal{K}, \\
  \mathcal{C}_2:\mathcal{I}\{\mathbf{h}_j\mathbf{x}[m]\}\unlhd\\\sigma_z\sqrt{\gamma_j}(\mathcal{I}\{\zeta\}\mathcal{R}\{{d}_j[n]\} +\mathcal{I}\{\zeta\}\mathcal{R}\{{d}_j[n]\}) ,\forall j\in \mathcal{K}.\\
  \end{cases}
  \end{eqnarray}
  From this reformulation, we can find that $\mathbf{x}[m]=\zeta\mathbf{x}[n]$ if $\zeta=1,-1,j,-j$.
  \end{proof}
  
 This means that when the data vectors are co-linear, the output vectors are also co-linear.  
 Based on this remark, we can reduce the number of output vector calculations or constraints per channel realization. However, this requires to find a subset of data vectors from which we can generate the remaining co-linear data vectors. We can propose the following data vectors as 
  \begin{eqnarray}
  \label{reduced combination}
  &	\mathcal{F}&=\Big\{\mathbf{d}[n]=[d_1[n], \hdots, d_K[n]]^T\Big|\\\nonumber &d_1[n]&\in \mathcal{Q}_1, d_j[n]\in \mathcal{Q}_m, \forall m\in\{1,2,3,4\}, \forall j\in \mathcal{K}/1\Big\},
  \end{eqnarray}
  where $\mathcal{Q}_m$ denotes the quadrant $m$. This set is not unique and different sets can be found using the same logic. The rest of data vectors can be represented using the set of selected data vectors set $\mathcal{F}$ as
  \begin{eqnarray}
  \label{reduced_data_vectors}
  	\mathbf{d}[m]=\zeta \mathbf{d}[n],\quad \zeta\in \{-1, j, -j\}, \quad\forall \mathbf{d}[n] \in \mathcal{F}.
  \end{eqnarray} 
  
 The main idea is to find the optimal output vectors for this subset of selected data vectors $\mathcal{F}$ using the optimization problem in \eqref{SLP_conventional}\eqref{SLP_conventional_temporal} or Algorithm 1. Then, we can find the output vectors for the remaining data vectors by exploiting the co-linearity as 
    \begin{eqnarray}
  \label{relation}
 \hspace{-0.9cm} \mathbf{x}[m]=\Big\{\zeta\mathbf{x}[n]:\zeta\in\{-1,-j,j\}\text{ if } \mathbf{d}[m]=\zeta\mathbf{d}[n], \forall \mathbf{d}[n] \in \mathcal{F} \Big\}.
  \end{eqnarray}
        Instead of calculating the output vectors for $2^{\sum_jm_j}$ data vectors, we calculate for $2^{\sum_jm_j-2}$ or $2^{\sum_jm_j-1}$ data vectors depending on the adopted modulations. For example, we have two users and each one of them is allocated 4-QAM. The number of possible data vectors is 16, which results in solving either 16 optimization problems with 4 constraints or one optimization problem with 64 constraints. In this case it is sufficient to consider four data vectors as 
        \begin{eqnarray}\nonumber
        	\hspace{-0.8cm}\mathcal{F}=\Bigg\{\frac{1}{\sqrt{2}}\begin{bmatrix}
        1+1i \\1+1i
        	\end{bmatrix},\frac{1}{\sqrt{2}}\begin{bmatrix}
        	1+1i \\1-1i
        	\end{bmatrix},\frac{1}{\sqrt{2}}\begin{bmatrix}
        1+1i \\-1+1i
        	\end{bmatrix}
        	\frac{1}{\sqrt{2}}\begin{bmatrix}
        	1+1i \\-1-1i
        	\end{bmatrix}\Bigg\}.
        \end{eqnarray}
        This results in  solving one optimization problem with 16 constraints or 4 optimization problems with 4 constraints. The remaining data vectors can be found  \eqref{reduced_data_vectors} and their corresponding output vectors \eqref{relation}.
        This decreases the vectors and matrices size in the optimization problem in \eqref{one_dimension} to quarter, and hence less complexity and algorithm running time. From now on, we can use $\mathcal{N}$ as the full set generated from brute force of all possible input data vectors or the reduced set as generated in this section, without loss loss of generality.

\section{Joint SLP precoding and phase rotation optimization}
\label{Joint SLP}
The importance of the constellation rotation in the context of symbol-level precoding stems from the fact that we want to increase the chances of generating constructive interference to add up in the correct detection region without additional processing, and thus results in better performance. This rotation is optimized per coherence time and it is fixed and does not change with each symbol transmission since this leads to problems in phase locking  and phase noise at the receiver side and additional signalling at the transmitter side. 
\subsection{Spatio-temporal presentation}
The spatio-temporal model helps simplifying the symbol-level problem; instead of solving $\mathcal{N}$ optimizations with $2K$ constraints, we can solve one optimization problem $2\mathcal{K}\mathcal{N}$  \cite{Alodeh_ICASSP2017,spano_twc,Spano_ICASSP2018}. This representation can capture the relation between the different transmitted data vectors in certain coherence time and helps optimizing parameters at the symbol-level or at the block-level. For example, we can optimize the constellation rotation for each user, this rotation is common among all the symbols transmitted in the coherence time/frame. The equivalent channel that captures the spatial and temporal domain can be formulated as
\begin{eqnarray}
\label{Spatio-temporal}
\mathbf{G}_j=\mathbf{h}_j\otimes \mathbf{I},
\end{eqnarray}
where $\mathbf{I}$ is an identity matrix of size $\mathcal{N}\times\mathcal{N}$, and $\mathbf{g}_j[n]$ is $n^{\text{th}}$ row of the matrix $\mathbf{G}_j$, and represents the channel at the data vector $\mathbf{d}[n]$. This formulation is required in the next subsection.


\subsection{SLP with constellation rotation optimization}

In order to find the optimal phase by which user's constellation should be rotated, we need to optimize the phase considering all the precoding vectors for all input data vectors. Therefore, we need to brute force all the symbols for all users, the number of symbol combination equals to $2^{\sum_{j\in K}m_j}$. The joint design of transmit output vectors and constellation rotation phases can then be posed
as the problem of minimizing the total radiated power (or maximizing the energy efficiency) subject
to meeting prescribed signal-to-noise ratio
(SNR) constraints at each of the receivers as
\begin{eqnarray}\nonumber
\label{SLP_rotation_1}
&\hspace{-0.5cm}\mathbf{x}_{opt}[n]& =\arg\underset{\mathbf{x}[n],\theta_j}{\min}\quad \frac{1}{N}\sum^{N}_{n=1}\|\mathbf{x}[n]\|^2\\ 
&\hspace{-0.5cm}s.t.& \begin{cases}\mathcal{C}_1:\mathcal{R}\{\mathbf{h}_j\mathbf{x}[n]\}-\sigma_z\sqrt{\gamma_j}\mathcal{R}\{{d}_j[n]\exp(-\theta_ji)\}\unlhd {0},\\ \forall j\in \mathcal{K}, \forall n \in \mathcal{N}\\
\mathcal{C}_2:\mathcal{I}\{\mathbf{h}_j\mathbf{x}[n]\}-\sigma_z\sqrt{\gamma_j}\mathcal{I}\{{d}_j[n]\exp(-\theta_ji)\}\unlhd {0},\\ \forall j\in \mathcal{K}, \forall n \in \mathcal{N}.\\
\end{cases}.
\end{eqnarray}
In \eqref{SLP_rotation_1},  the temporal dimension is added to optimize the phase rotation considering all possible data vectors or the reduced set of data vectors, this means that the phase rotation for each user does not change with each symbol and is fixed as long as CSI or the set of served users does not change. Each user has a specific constellation rotation; $\theta_j$ is the phase rotation for the user $j$ and does not change with respect to $n$ and it is optimized once when the channel changes unlike the output vector $\mathbf{x}[n]$ that changes with each symbol.  For the sake of convenience, we reformulate the optimization problem by gathering the variable that we want to optimize in one side as  
\begin{eqnarray}\nonumber
&\hspace{-0.5cm}\mathbf{x}[n]& =\arg\underset{\mathbf{x}[n],\theta_j}{\min}\quad \frac{1}{N}\sum^{N}_{n=1}\|\mathbf{x}[n]\|^2\\
&\hspace{-0.5cm}s.t.& \begin{cases}\mathcal{C}_1:\mathcal{R}\{\exp(\theta_ji)\mathbf{h}_j\mathbf{x}[n]\}-\sigma_z\sqrt{\gamma_j}\mathcal{R}\{{d}_j[n]\}\unlhd {0},\\ \forall j\in \mathcal{K}, \forall n \in \mathcal{N}\\
\mathcal{C}_2:\mathcal{I}\{\exp(\theta_ji)\mathbf{h}_j\mathbf{x}[n]\}-\sigma_z\sqrt{\gamma_j}\mathcal{I}\{{d}_j[n]\}\unlhd {0},\\ \forall j\in \mathcal{K}, \forall n \in \mathcal{N}.\\
\end{cases}
\end{eqnarray}
A further reformulation using $t_j=\exp(\theta_ji)$ can be described as
\begin{eqnarray}\nonumber
\label{one_dimension}
&\hspace{-0.8cm}\mathbf{x}[n]& =\arg\underset{\mathbf{x}[n],u_j}{\min}\quad \frac{1}{N}\sum^{N}_{n=1}\|\mathbf{x}[n]\|^2\\
&\hspace{-0.8cm}s.t.&\quad \begin{cases}\mathcal{C}_1:\mathcal{R}\{t_j\mathbf{h}_j\mathbf{x}[n]\}-\sigma_z\sqrt{\gamma_j}\mathcal{R}\{{d}_j[n]\}\unlhd {0},\\ \forall j\in \mathcal{K}, \forall n \in \mathcal{N}\\
\mathcal{C}_2:\mathcal{I}\{t_j\mathbf{h}_j\mathbf{x}[n]\}-\sigma_z\sqrt{\gamma_j}\mathcal{I}\{{d}_j[n]\}\unlhd {0},\\ \forall j\in \mathcal{K}, \forall n \in \mathcal{N}\\
\mathcal{C}_3:|t_j|=1,\forall j \in K.
\end{cases}
\end{eqnarray}
The set of constraints $\mathcal{C}_1$, $\mathcal{C}_2$ and $\mathcal{C}_3$ are non-convex, $\mathcal{C}_1$ and $\mathcal{C}_2$ are bilinear, $C_3$ is a constant modulus constraint that should be satisfied with equality and cannot be relaxed to $|t_j|\leq 1$ since it results in $t_j=0$, and consequently this means no transmission will occur.
 To solve the problem, it is useful to utilize the two dimensional channel model in \eqref{Spatio-temporal} to facilitate the effect of different input data vectors on finding the optimal phase rotation. This model makes the optimization problem in \eqref{one_dimension} more compact and easier to tackle. Let us define $\mathbf{p}\in\mathbb{C}^{NM\times 1}$ as the vector  that stacks all the output vectors $\mathbf{x}[n]$ for all (or reduced) input data vectors as
\begin{eqnarray}
 \mathbf{p}=\begin{bmatrix}
  \mathbf{x}[1]\\
  \mathbf{x}[2]\\
  \vdots\\
  \mathbf{x}[N]\\ 
 \end{bmatrix},
\end{eqnarray}
and using the spatio-temporal channel is defined in \eqref{Spatio-temporal}, the problem can be expressed as
\begin{eqnarray}\nonumber
\label{phaserotation}
&\vspace{-0.5cm}\mathbf{p}_{opt}&=\arg\underset{\mathbf{p},t_j}{\min}\quad \frac{1}{N}\|\mathbf{p}\|^2\\
&s.t.&\quad \begin{cases}\mathcal{C}_1:\mathcal{R}\{\mathbf{g}_j[n]\mathbf{p}_j\}-\sigma_z\sqrt{\gamma_j}\mathcal{R}\{{d}_j[n]\}\unlhd {0},\forall j\in \mathcal{K},\\
\mathcal{C}_2:\mathcal{I}\{\mathbf{g}_j[n]\mathbf{p}_j\}-\sigma_z\sqrt{\gamma_j}\mathcal{I}\{{d}_j[n]\}\unlhd {0}, \forall j\in \mathcal{K},\\
\mathcal{C}_3:\mathbf{p}_j=t_j\mathbf{p},\forall j \in K\\
\mathcal{C}_4:|t_j|=1,\forall j \in K,
\end{cases}
\end{eqnarray}
where $\mathbf{g}_j[n]$ is the $n^{\text{th}}$ row of the matrix $\mathbf{G}_j$ defined in \eqref{Spatio-temporal} and is the equivalent channel for user $j$ at the symbol instant $n$, which takes into consideration the possibility of all symbols combinations that are transmitted and ISI if it exists. 
A common approach to deal with non-convex problems in
practice is to relax the non-convex constraints $\mathcal{C}_1$-$\mathcal{C}_3$ to obtain a convex
problem that represents the original one.
To solve the problem, we use semidefinite programming (SDP), we define the vector $\mathbf{v}\in \mathbb{C}^{(\mathcal{N}M+K)\times 1}$ as a vector that includes all the parameters that we want to optimize as
\begin{eqnarray}
\mathbf{v}=\begin{bmatrix}\mathbf{p}\\\mathbf{t}\end{bmatrix}=\begin{bmatrix}\mathbf{p} \\ t_1\\ \vdots\\ t_k\end{bmatrix}.
\end{eqnarray}
Utilizing the vector $\mathbf{v}$ helps us get rid of the bilinear constraint $\mathcal{C}_3$.
For the ease of exposition, we define the matrices  $\mathbf{P}\triangleq\mathbf{p}\mathbf{p}^H$, $\mathbf{T}\triangleq\mathbf{t}\mathbf{t}^H$, and the block matrix $\mathbf{V}\triangleq\mathbf{v}\mathbf{v}^H$ as
\begin{eqnarray}
\mathbf{V}=\begin{bmatrix}
\mathbf{P}&\mathbf{F}\\
\mathbf{F}^H&\mathbf{T}
\end{bmatrix},
\end{eqnarray}
where  $\mathbf{F}$ is formulated as
\begin{eqnarray}\nonumber
\mathbf{F}&=&\begin{bmatrix}
\mathbf{p}_1&\mathbf{p}_2&\hdots&\mathbf{p}_K
\end{bmatrix}.
\end{eqnarray}
In case of optimal solution, the rank of the matrix $\mathbf{F}$ should equal to one, to enforce the columns to be linearly dependent. To guarantee that,
the matrix $\mathbf{T}$ should have the following structure to satisfy the constraint $\mathcal{C}_3,\mathcal{C}_4$ in \eqref{phaserotation}
\begin{eqnarray}
\mathbf{T}=\begin{cases} [\mathbf{T}]_{i,j}=1, &\quad  i=j\\ |[\mathbf{T}]_{i,j}|=1, &\quad i\neq j \end{cases}.
\end{eqnarray}
Using semidefinite relaxation, the optimization problem can be formulated as 

\begin{eqnarray}\nonumber
\label{semidefinite_norm}
&\mathbf{P}_{opt}&=\arg\underset{\mathbf{P}}{\min}\quad tr(\mathbf P)\\
&s.t.& \begin{cases}\mathcal{C}_1:\mathcal{R}\{\mathbf{g}_j[n]\mathbf{p}_j\}-\sigma_z\sqrt{\gamma_j}\mathcal{R}\{{d}_j[n]\}\unlhd {0}, \forall j\in \mathcal{K},\\
\mathcal{C}_2:\mathcal{I}\{\mathbf{g}_j[n]\mathbf{p}_j\}-\sigma_z\sqrt{\gamma_j}\mathcal{I}\{{d}_j[n]\}\unlhd {0}, \forall j\in \mathcal{K},\\
\mathcal{C}_3:\mathbf{V}\succcurlyeq 0,\\
\mathcal{C}_4:[\mathbf{T}]_{i,j}=1, i=j, \\
\mathcal{C}_5:|[\mathbf{T}]_{i,j}|=1, i\neq j.
\end{cases}
\end{eqnarray}

\begin{proposition}
The equality constant modulus constraint $\mathcal{C}_5$ guarantees that rank of matrix $\mathbf{P}$ equals to 1.  
\end{proposition}
\begin{proof}
In the case of optimal solution, the constant modulus constraint guarantees $|v_j|=1$ in \eqref{one_dimension} and hence that the vectors in $\mathbf{F}$ are linearly dependent. This results in the rank of $ \mathbf{P}$ equals to 1, and consequently the rank of $\mathbf{P}\equiv\mathbf{F}\mathbf{F}^H$ is one.  
\end{proof}
\subsection{Detection Region}
\label{detection region}
In \cite{maha_twc_2}, we illustrate the detailed constraints of the circular and rectangular modulations to facilitate the design of constructive interference in SLP context. For the sake of completion, we revisit them again in the paper to make it clear for the reader, and clarify them in the optimization problem \eqref{semidefinite_ac}.
\begin{figure}[h]
	\begin{center}
		\includegraphics[scale=0.4]{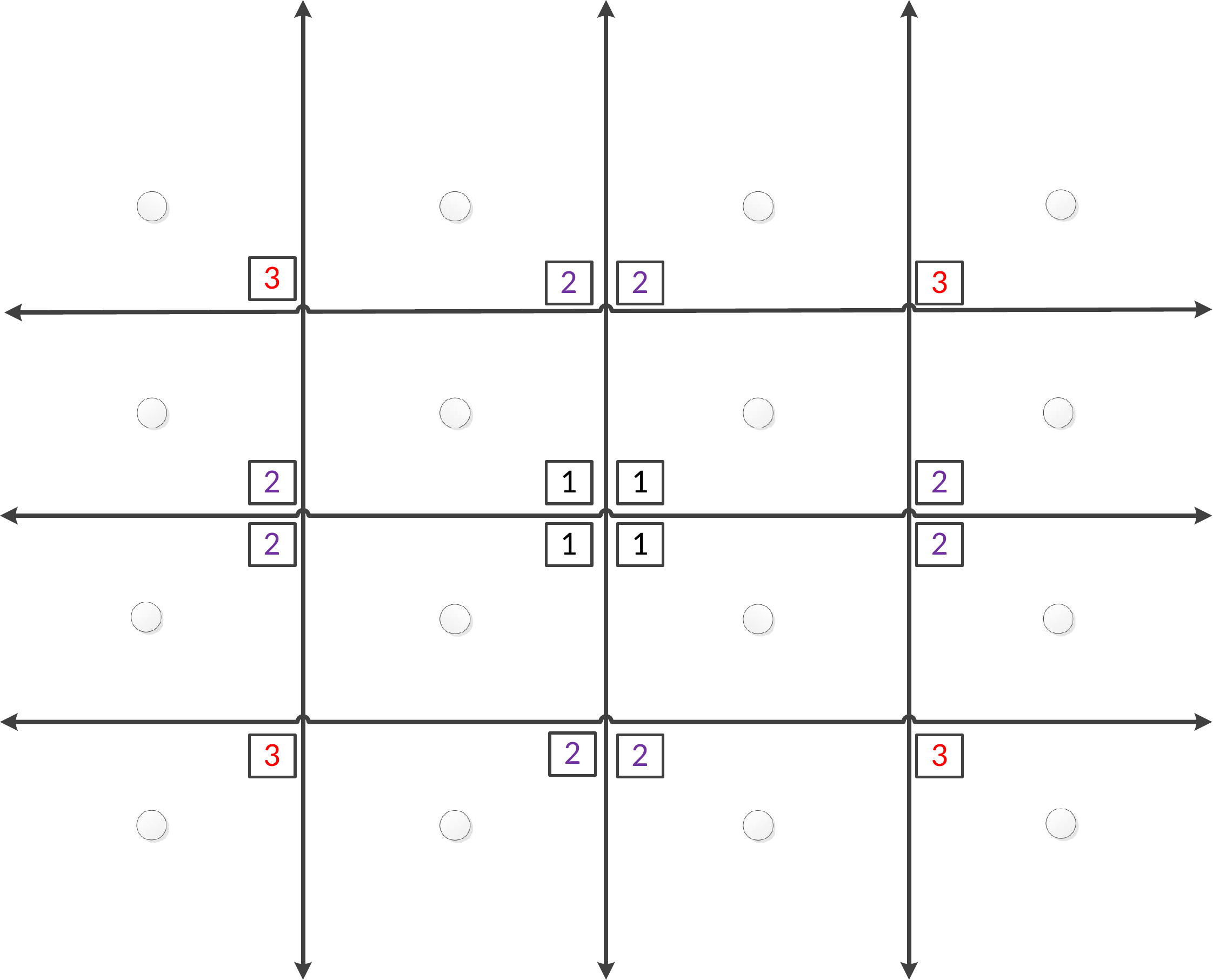}
		\caption{\label{qam} 16 QAM Constellation}
	\end{center}
\end{figure}

\subsubsection{Detailed Optimization For Circular Modulation}
For any circular modulation ($M$-PSK, APSK), the optimization can be expanded as following
\begin{eqnarray}\nonumber
\label{SLP_APSK}
\hspace{-0.3cm}&\mathbf{P}_{opt}&=\arg\min_\mathbf{P}\quad tr(\mathbf{P})\\ &s.t.&\quad\begin{cases}  
\mathcal{C}_1:\mathcal{R}\{\mathbf{g}_j[n]\mathbf{p}_j\}\unlhd \sqrt{\gamma_j}\sigma_z\mathcal{R}\{d_j[n]\}\\
\mathcal{C}_2:\mathcal{I}\{\mathbf{g}_j[n]\mathbf{p}_j\}\unlhd \sqrt{\gamma_j}\sigma_z\mathcal{I}\{d_j[n]\} \forall j\in K\\
\mathcal{C}_3:\mathcal{I}\{\mathbf{g}_j[n]\mathbf{p}_j\}-a_j[n]\mathcal{R}\{\mathbf{g}_j[n]\mathbf{p}_j\}=0\\
\mathcal{C}_4:\mathbf{V}\succcurlyeq 0,\\
\mathcal{C}_5:[\mathbf{T}]_{i,j}=1, i=j, \\
\mathcal{C}_6:|[\mathbf{T}]_{i,j}|=1, i\neq j.
\end{cases}
\end{eqnarray}
where
\begin{eqnarray}\nonumber
\mathcal{R}\{\mathbf{g}_j[n]\mathbf{p}_j\}&=&\frac{\mathbf{g}_j[n]\mathbf{p}_j+\mathbf{p}^H_j\mathbf{g}^H_j[n]}{2},\\\nonumber
\mathcal{I}\{\mathbf{g}_j[n]\mathbf{p}_j\}&=&\frac{\mathbf{g}_j[n]\mathbf{p}_j-\mathbf{p}^H_j\mathbf{g}^H_j[n]}{2i},\\\nonumber
a_j[n]&=&\tan(\angle d_j[n]).	
\end{eqnarray}
$\mathcal{C}_1$, $\mathcal{C}_2$ and $\mathcal{C}_3$  are formulated to guarantee that the received signal lies in the correct detection region, which depends on the data symbols. A detailed formulation for $\mathcal{C}_1$, $\mathcal{C}_2$ can be expressed as

\begin{itemize}
	\item For the inner-constellation symbols, the constraints $\mathcal{C}_1$, $\mathcal{C}_2$ should guarantee that the received signals achieve the exact constellation point. The constraints can be written as:    
	\begin{eqnarray}
	\label{inner_circular}
	&\mathcal{C}_1:&\mathcal{R}\{\mathbf{g}_j[n]\mathbf{p}_j\}=\sigma_z\sqrt{\gamma_j}\mathcal{R}\{d_j[n]\}\\\nonumber
	&\mathcal{C}_2:&\mathcal{I}\{\mathbf{g}_j[n]\mathbf{p}_j\}=\sigma_z\sqrt{\gamma_j}\mathcal{I}\{d_j[n]\}.
	\end{eqnarray}
	$\mathcal{C}_3$ is not required in this case.
	\normalsize
	\item  Outermost constellation symbols, the constraints $\mathcal{C}_1$, $\mathcal{C}_2$ should guarantee the received signals lie in the correct detection, which is more flexible than the inner constellation points. The constraints can be written as:     
	\begin{eqnarray}\nonumber
	\label{outermost_circular}
	&\hspace{-0.7cm}\mathcal{C}_1:&\mathcal{R}\{\mathbf{g}_j[n]\mathbf{p}_j\}\geq\sigma_z\sqrt{\gamma_j}\mathcal{R}\{d_j[n]\}, \mathcal{R}\{d_j[n]\}\geq 0 \\\nonumber
	&\quad&\mathcal{R}\{\mathbf{g}_j[n]\mathbf{p}_j\}\leq\sigma_z\sqrt{\gamma_j}\mathcal{R}\{d_j[n]\}, \mathcal{R}\{d_j[n]\}\leq 0 \\
	&\hspace{-0.7cm}\mathcal{C}_2:&\mathcal{I}\{\mathbf{g}_j[n]\mathbf{p}_j\}\geq\sigma_z\sqrt{\gamma_j}\mathcal{I}\{d_j[n]\}, \mathcal{I}\{d_j[n]\}\geq 0\\\nonumber
	&\quad&\mathcal{I}\{\mathbf{g}_j[n]\mathbf{p}_j\}\leq\sigma_z\sqrt{\gamma_j}\mathcal{I}\{d_j[n]\}, \mathcal{I}\{d_j[n]\}\leq 0.
	\end{eqnarray} 
	$\mathcal{C}_3$ guarantees that the received symbol has a certain phase. It should be clear that $\tan(\dot)$ cannot preserve the sign. Therefore, $\mathcal{C}_1,\mathcal{C}_2$, and $\mathcal{C}_3$ should be used together. 
\end{itemize}

\subsubsection{Detailed Optimization For Rectangular Modulation}
For rectangular modulation (e.g. $M$-QAM), the previous optimization can be simplified as:  
\begin{eqnarray}\nonumber
\label{SLP_QAM}
\hspace{-0.5cm}&\mathbf{P}_{opt}&=\arg\min_\mathbf{P} \quad tr(\mathbf{P})\\ \hspace{-0.5cm}&s.t.&\quad\begin{cases}  
\mathcal{C}_1:\mathcal{R}\{\mathbf{g}_j[n]\mathbf{p}_j\}\unlhd \sqrt{\gamma_j}\sigma_z\mathcal{R}\{d_j[n]\}, \forall j\in K, \forall n\in \mathcal{N}\\
\mathcal{C}_2:\mathcal{I}\{\mathbf{g}_j[n]\mathbf{p}_j\}\unlhd \sqrt{\gamma_j}\sigma_z\mathcal{I}\{d_j[n]\}, \forall j\in K, \forall n\in \mathcal{N}\\
\mathcal{C}_3:\mathbf{V}\succcurlyeq 0,\\
\mathcal{C}_4:[\mathbf{T}]_{i,j}=1, i=j, \\
\mathcal{C}_5:|[\mathbf{T}]_{i,j}|=1, i\neq j.
\end{cases}
\end{eqnarray}
$\mathcal{C}_1$, $\mathcal{C}_2$  are formulated to guarantee that the received signal lies in the correct detection region, which depends on the data symbols. A detailed formulation for $\mathcal{C}_1$, $\mathcal{C}_2$ can be expressed as

\begin{itemize}
	\item For the inner-constellation symbols, $\mathcal{C}_1$, $\mathcal{C}_2$ can be formulated as \eqref{inner_circular}.
	
	\item Outer constellation symbols, the constraints $\mathcal{C}_1$, $\mathcal{C}_2$ should guarantee the received signals lie in the correct detection.  The constraints can be written as:     
	\begin{eqnarray}\nonumber
	&\mathcal{C}_1:&\mathcal{R}\{\mathbf{g}_j[n]\mathbf{p}_j\}\geq\sigma_z\sqrt{\gamma_j}\mathcal{R}\{d_j[n]\}, \mathcal{R}\{d_j[n]\}\geq 0 \\\nonumber
	&\quad&\mathcal{R}\{\mathbf{g}_j[n]\mathbf{p}_j\}\leq\sigma_z\sqrt{\gamma_j}\mathcal{R}\{d_j[n]\}, \mathcal{R}\{d_j[n]\}\leq 0 \\
	&\mathcal{C}_2:&\mathcal{I}\{\mathbf{g}_j[n]\mathbf{p}_j\}=\sigma_z\sqrt{\gamma_j}\mathcal{I}\{d_j[n]\}.
	\end{eqnarray}
	\begin{eqnarray}\nonumber
	&\mathcal{C}_1:&\mathcal{R}\{\mathbf{g}_j[n]\mathbf{p}_j\}=\sigma_z\sqrt{\gamma_j}\mathcal{R}\{d_j[n]\}\\\nonumber
	&\mathcal{C}_2:&\mathcal{I}\{\mathbf{g}_j[n]\mathbf{p}_j\}\geq\sigma_z\sqrt{\gamma_j}\mathcal{I}\{d_j[n]\}, \mathcal{I}\{d_j[n]\}\geq 0\\\nonumber
	&\quad&\mathcal{I}\{\mathbf{g}_j[n]\mathbf{p}_j\}\leq\sigma_z\sqrt{\gamma_j}\mathcal{I}\{d_j[n]\}, \mathcal{I}\{d_j[n]\}\leq 0.
	\end{eqnarray}
	\item Outermost constellation symbols, the constraints $\mathcal{C}_1$, $\mathcal{C}_2$ should guarantee the received signals lie in the correct detection. The constraints can be formulated as \eqref{outermost_circular}.    
	
\end{itemize}
In the next section, we replace the norm constraint $\mathcal{C}_5$ by a set of affine constraints that converts \eqref{semidefinite_norm} into convex optimization problem.

\subsection{Argument cuts for constant modulus constraint}

Back to our problem, the constraint $\mathcal{C}_5$ is not convex, but we will show that the
problem still can be efficiently solved using convex optimization. 
Semidefinite relaxation can only be utilized to obtain a lower bound to the
optimal objective function and potentially determine an approximate solution to the original problem. The semidefinite relaxation has several nice properties and
is typically solved using a primal-dual interior point method,
for which there are several standard optimization tools, e.g.,
SeDuMi.

The constant modulus constraint can be reformulated by a set of affine constraints using argument cuts defined  in \cite{Chinese_norm}. The argument cuts were proposed to solve the problem of precoding design for single group multicasting globally and to tackle the rank one constraint imposed by the semidefinite relaxation. Here, we use the same concept of the argument cut to tackle the constant modulus constraints and relax them to solve the problem efficiently. To do so, we reformulate $\mathcal{C}_5$ as
\begin{eqnarray}
[\mathbf{T}]_{i,j}=c_{i,j}, \text{where},\{c_{i,j} \in \mathbb{C} | |c_{i,j}| = 1\}.
\end{eqnarray}
We need to find a tight convex relaxation for $\mathcal{C}_5$, we introduce $u_{i,j} = \mathcal{R}(c_{i,j})$ and $w_{i,j} = \mathcal{I}(c_{i,j})$, and assume
that the argument of $c_{i,j}$ satisfies $\arg (c_{i,j}) \in [l_{i,j}, u_{i,j}]$. Let
\begin{eqnarray}\nonumber
  \hspace{-0.4cm}  &\mathcal{D}[l_{i,j},u_{i,j}]& =\Bigg\{(w_{i,j}, v_{i,j})\bigg|\begin{array}{c} c_{i,j}= w_{i,j} + v_{i,j}i,|c_{i,j}| \geq 1\\\arg (c_{i,j}) \in [l_{i,j}, u_{i,j}] \end{array} \Bigg\},
 \end{eqnarray}
and let ${Conv}(\mathcal{D}[l_{i,j},u_{i,j}])$ be the convex envelope of the set
$\mathcal{D}[l_{i,j},u_{i,j}]$. The following proposition characterizes $Conv(\mathcal{D}[l_{i,j},u_{i,j}])$ in
the general case if $u_{i,j} - l_{i,j} \leq \pi$.
\begin{proposition}{\cite{Chinese_norm}}
 Suppose $l_{i,j}$ and $u_{i,j}$ in satisfying $u_{i,j} - l_{i,j} \leq \pi$. Then, 
\begin{eqnarray}\nonumber
\label{convex set}
 \hspace{-0.9cm}Conv(\mathcal{D}[l_{i,j},u_{i,j}])
 =\Bigg\{(w_{i,j},v_{i,j})\Bigg|
 \begin{array}{c}
\sin(l_{i,j})w_{i,j} - \cos(l_{i,j})v_{i,j} \leq 0\\
\sin(u_{i,j})w_{i,j} - \cos(u_{i,j})v_{i,j} \geq 0\\
a_{i,j}w_{i,j} + b_{i,j}v_{i,j} \geq  a^2_{i,j}+b^2_{i,j},
\end{array}
\end{eqnarray}
where 
\begin{eqnarray}
\label{convex set 2}
a_{i,j}=\frac{\cos{l_{i,j}}+\cos{u_{i,j}}}{2},\quad b_{i,j}=\frac{\sin{l_{i,j}}+\sin{u_{i,j}}}{2}.
\end{eqnarray}	
\end{proposition}The proposition convexifies the set $\mathcal{D}[l_{i,j},u_{i,j}]$ by representing the constant modulus and bilinear constraints using three linear inequalities in \eqref{convex set}. This helps solving the problem by tightening the the width
of the interval $[l_{i,j}, u_{i,j}]$, the tighter the constraint becomes as the width of the
interval goes to zero. 

With the help of the argument
cuts, we are able to develop efficient convex relaxations for
problem. Using \eqref{convex set}, the optimization problem can be formulated as
\begin{eqnarray}\nonumber
\footnotesize
\label{semidefinite_ac}
\hspace{-0.5cm}\mathbf{P}_{opt}&=&\arg\underset{\mathbf{P}}{\min}\quad tr(\mathbf P)\\
&s.t.& \begin{cases}\mathcal{C}_1:\mathcal{R}\{\mathbf{g}_j[n]\mathbf{p}_j\}-\sigma_z\sqrt{\gamma_j}\mathcal{R}\{{d}_j[n]\}\unlhd {0}, \forall j\in \mathcal{K},\\
\mathcal{C}_2:\mathcal{I}\{\mathbf{g}_j[n]\mathbf{p}_j\}-\sigma_z\sqrt{\gamma_j}\mathcal{I}\{{d}_j[n]\}\unlhd {0}, \forall j\in \mathcal{K},\\
\mathcal{C}_3:\mathbf{V}\succcurlyeq 0,\\
\mathcal{C}_4:[\mathbf{T}]_{i,j}=1, i=j, \\
\mathcal{C}_5:[\mathbf{T}]_{i,j}=w_{i,j}+v_{i,j}\iota, i\neq j\\
\mathcal{C}_6:\sin(l_{i,j})w_{i,j}-\cos(l_{i,j})v_{i,j}\leq 0, i\neq j\\
\mathcal{C}_7:\sin(u_{i,j})w_{i,j}-\cos(u_{i,j})v_{i,j}\geq 0, i\neq j\\
\mathcal{C}_8:a_{i,j}w_{i,j}+b_{i,j}v_{i,j}\geq a^2_{i,j}+b^2_{i,j}, i\neq j.\\
\end{cases}
\end{eqnarray}
 The optimization problem in \eqref{semidefinite_ac} is convex semidefinite program, and can be used to solve the optimization in \eqref{semidefinite_norm} by iterating through branch and bound algorithm as it will be described in the next section. The solution for the optimization problem results in $[\mathbf{U}]_{j,i}=[\mathbf{U}]^{*}_{i,j}$ since the matrix $\mathbf{U}$ is semidefinite one. All constraints in $\mathcal{C}_5$-$\mathcal{C}_8$ are expanded  into $K-1$ constraints. 

\section{Algorithm for Joint Constellation rotation and symbol-level precoding optimization }
\label{Algorithm}
In this section, we present an algorithm to solve the optimization problems \eqref{semidefinite_ac}. Although we reformulate the constraints into relaxed semidefinite program, an enumeration procedure to search for the optimal solution is required.
In this subsection, a branch-and-bound algorithm for globally solving problem \eqref{semidefinite_norm} is proposed. The main idea is to relax the original problem, with appropriate set of parameter and gradually 
tighten the relaxation by reducing the width of the associated
intervals. 

The main purpose of using branch and bound is to reformulate the problem as a tree-search problem for which we can prune large parts of the tree in order to reduce the
computational complexity. The general concept of branch and bound algorithm is based on dividing the feasible region into smaller sub-regions
and constructs sub-problems over them
recursively. In the enumeration procedure, a lower bound for
each sub-problem is calculated by solving a relaxation problem.
A sub-problem with a lower bound being greater than the obtained
upper bound is considered to be inactive, which the global solution does not belong in
its feasible region, and therefore will no longer be divided into subgroup and consequently there will be no futher enumeration.
The procedure terminates until all active sub-problems have
been enumerated, and then an optimal solution within a given
relative error tolerance can be obtained. There are two aspects to be considered while branch and bound algorithm is applied: 1) The optimum objective
function value in each problem node represents a lower bound
to its descendent nodes.  2) Certain nodes and their corresponding subtrees are
pruned from the search tree  if the subproblem
is infeasible, the objective function value of the subproblem exceeds
the current upper bound or the other branches lower bounds are higher. In our problem, the lower and the upper bound can be describe as
\begin{itemize}
	\item Upper-bound: the upperbound for the problem is obtained when no rotation is applied $c_{i,j}=1$, the transmit power is the highest in this scenario.
	\item Lower-bound: the lowerbound is the optimum objective function value from its asendending nodes
	\item Termination: The problem is terminated if the the gap between the upper and lower bound is lower than preselected error value.
\begin{eqnarray}
\frac{U-L}{L}\leq \epsilon
\frac{K-\sum|c_{i,j}|}{K}\leq \epsilon_\circ.
\end{eqnarray}
\end{itemize}

\begin{algorithm}
	\SetAlgoLined
	\textbf{Input} $\mathbf{h}_j$, $\mathbf{D}$, $\epsilon_\circ$, $l_{i,j}=0$, $u_{i,j}=\pi$ \\
	\textbf{Output} $\mathbf{x}^{f}$, $\theta_k$\\
	initialization $\epsilon=10$, $m=0$,  $\mathcal{A}^0=\prod^{K}_{k=0}[0,\pi]$\\
	\begin{itemize}
		\item Solve $\mathcal{OP}(\mathcal{A}^0)$ in \eqref{semidefinite_ac} for its optimal solution ($\mathbf{P}$) and its optimal value, and consider as initial lower bound $L_0$.\\
		
		\item solve $\mathcal{OP}$ in \eqref{semidefinite_ac}, assuming $l_{i,j}=0,\forall j,i$ and $u_{i,j}=0,\forall j,i$ for its optimal solution, and consider it as an upperbound $U^*$.
		\item\textbf{loop}:if $(U^* - L_m)/L_m \leq \epsilon_\circ$  then\\
		return ($\mathbf{P}$) and terminate the algorithm.\\
	
		Set $m \leftarrow m + 1$.\\
		For each set $(i,j)$, Branch the region $[l_{i,j},u_{i,j}]$ into two subregions $[l^m_{i,j},t^m_{i,j}]$, $[t^m_{i,j},u^m_{i,j}]$, where  $t^m_{i,j}=(l^m_{i,j}+u^m_{i,j})/2$\\
		Solve the problems for the newly defined branches at the level $m$.\\
		Find the branch that gives the minimal Lower power bound and set it as the new lower bound $L_m$
	\end{itemize}	
	\caption{Branch and Bound for joint constellation rotation and precoding (SLPRo) }
\end{algorithm}

\normalsize

\subsection{Solution}
With $\mathbf{P}$ at hand, a straightforward way to find the
output vector $\mathbf{x}$ is by decomposing the matrix  $\mathbf{P}$ into
\begin{eqnarray}
\mathbf{P}=\lambda \mathbf{e}\mathbf{e}^H,	
\end{eqnarray}
where $\lambda$ denotes the largest eigenvalue of $\mathbf{P}$  and the only nonzero eigen value, and $\mathbf{e}$ denotes
the eigenvector of $\mathbf{P}$ corresponding to $\lambda$. Now, we can express the output vector $\mathbf{p}$ as:
\begin{eqnarray}
\mathbf{p}=\sqrt{\lambda}\mathbf{e},
\end{eqnarray}
and the optimal phase rotation for user $k$,$\phi_k$, can be found as 
\begin{eqnarray}
&\alpha_k&={\mathbf{e}^H\mathbf{p}_k},\\
&\phi_k&=\angle \alpha_k.
\end{eqnarray}
In order to ensure that the rank of $\mathbf{X}$ equals to 1, the value $|\alpha_k|$ should be equal to 1. if $|\alpha_k|\leq 1$, this means that rank of $\mathbf X$ is higher than one.
\begin{figure}[h]
	\label{model_slp}
	\vspace{-4.5cm}
	\hspace{-9.0cm}	\includegraphics[trim=0 0 0 50,clip,scale=1.0]{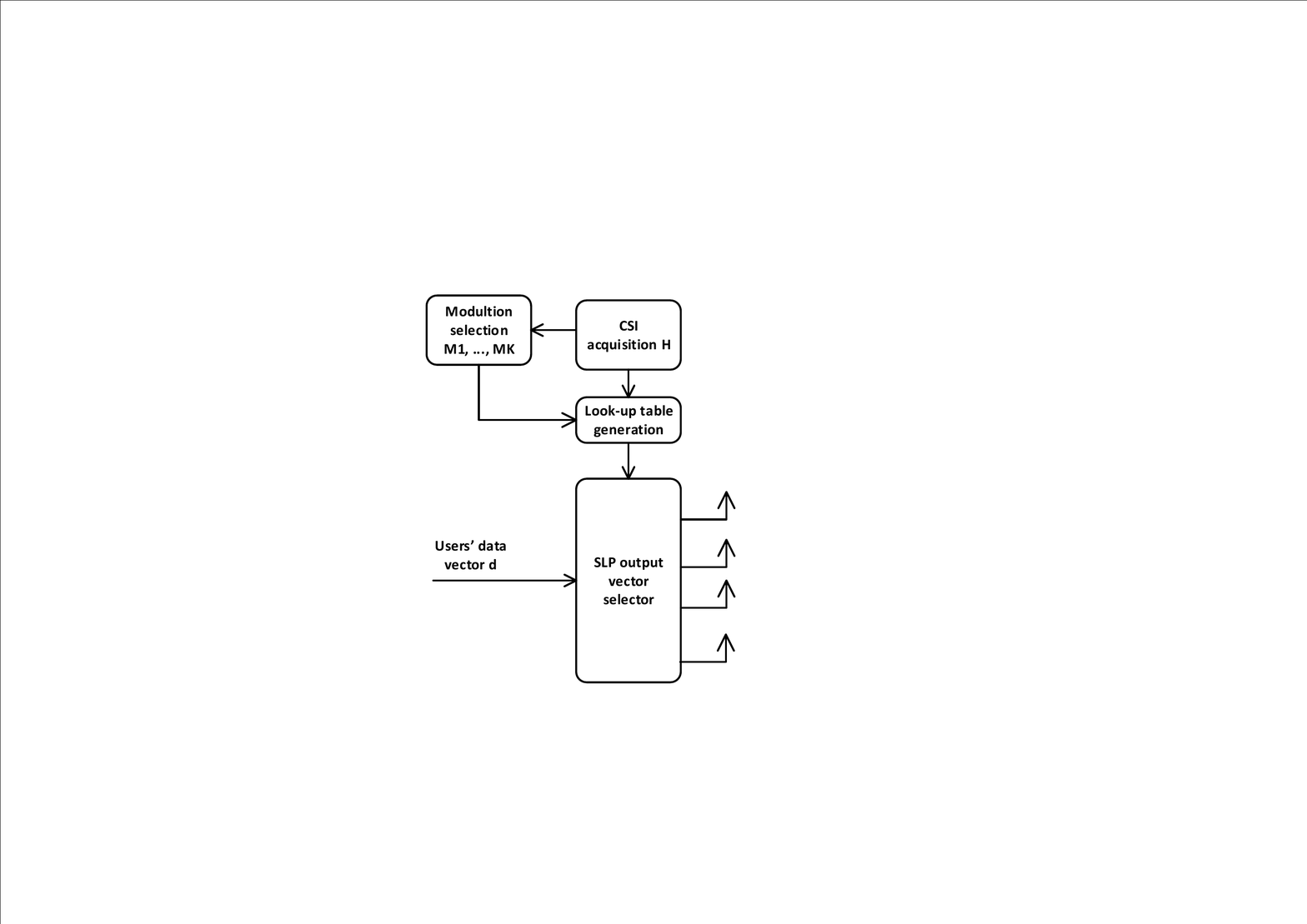}
	\vspace{-5.8cm}
	\caption{\label{model_slp} Proposed Symbol-level Precoding System}	
\end{figure}
\begin{remark}
It is worth noting that 
ML detection at user $k$ can be performed by simply
checking the phase of $\frac{\mathbf{y}_k}{\exp(-i\theta_k)}$ in case of MPSK modulation. For $M$-QAM modulation,  the detector finds the correct the detection region of $\frac{\mathbf{y}_k}{\exp(-i\theta_k)}$ as described in \ref{detection region}.
\end{remark}

\subsection{Generic Framework of SLP}
In this subsection, We propose a generic framework for efficient SLP system in the downlink of MISO system. In Fig. \ref{model_slp}, CSI and modulation allocation is fed to the look-up table generator. The look-up table changes with CSI. Once data vector $\mathbf{d}[n]$ needs to be precoded, the optimal output vector is selected. It can be urgued that in the case of no constellation rotation optimization, the output vector can be optimized immediately based on the input data vector, and then stored in the look-up table; as the same data vector may come later in the data streams.  

Regarding the look-up table generation as described in Fig. 3, we can summarize it in the following steps

\begin{itemize}
\item The channel state information of selected users and their allocated modulations are fed to the look-up table generator. A spatio-temporal channel is formulated as \eqref{Spatio-temporal}.
\item  Based on the selected modulation, the subset of the data vectors that can reproduce the full brute space is generated as \eqref{reduced combination}. This subset is fed to the optimization process. Also, the rotation decision is provided to know whether to optimize the phase rotation and the output vectors as in SLPRo algorithm or just output vectors as in SLP. 
\item  Mapping the selected data vectors and their corresponding output vectors is rendered. Finally, the full brute force space of data vectors is generated as in \eqref{reduced_data_vectors}. The output vectors for remaining data vectors are calculated as\eqref{relation}.
\end{itemize}

 \begin{figure*}[t]
	\hspace{0.2cm}
	
	\begin{tabular}[t]{c}
		\begin{minipage}{17 cm}
			\vspace{-1cm}
			\label{Framework}
			\hspace{-2.1cm}	\includegraphics[trim=020 50 50 50 ,clip,scale=0.525]{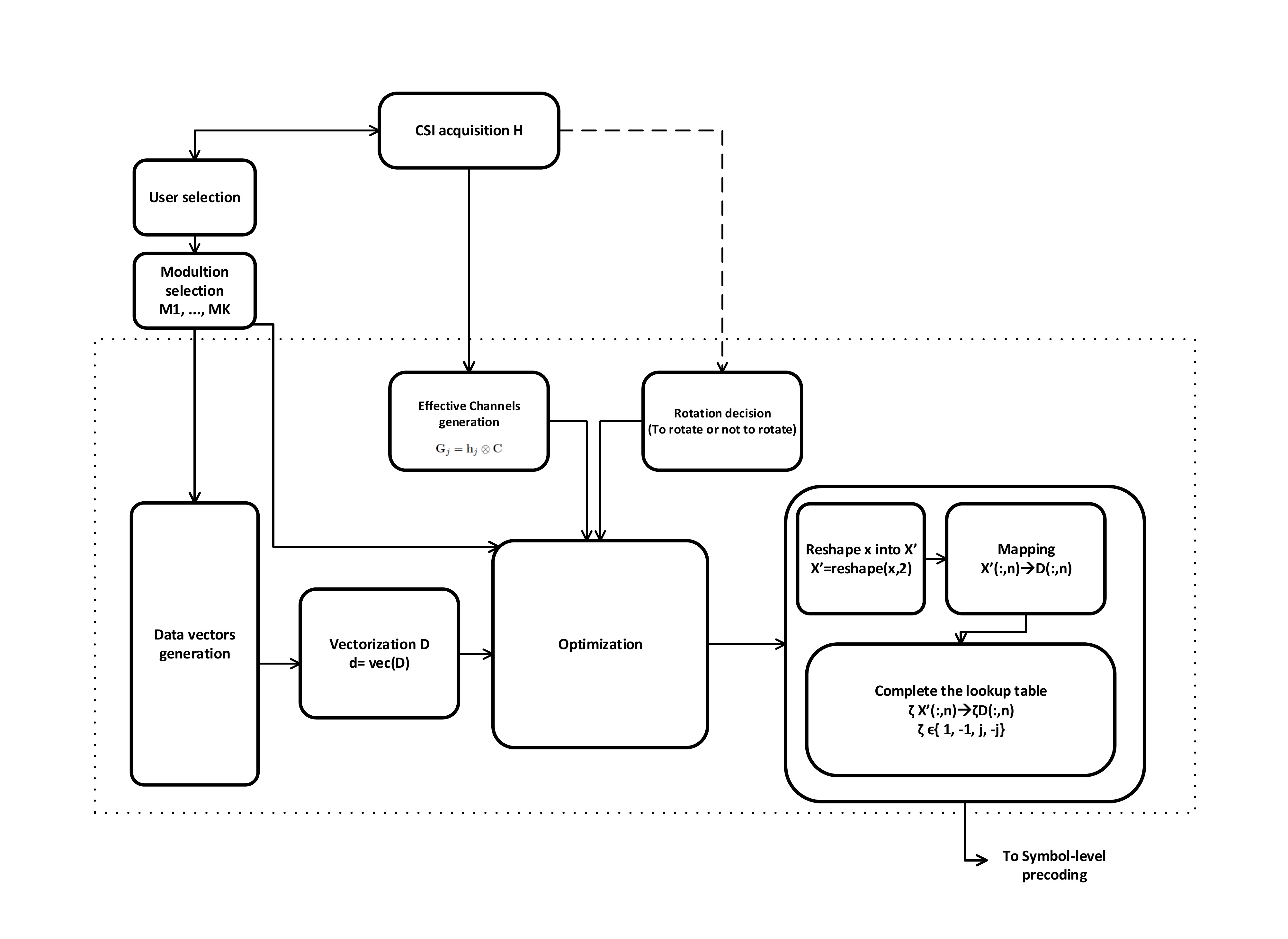}
			\caption{Schematic diagram for proposed framework of look-up table generation taking into the account the spatio-temporal model}	
		\end{minipage}
	\end{tabular}	
\end{figure*}

\section{Numerical Results}
\label{Results}
In order to assess the performance of the proposed transmissions schemes, Monte-Carlo simulations of the different algorithms have been conducted to
study the performance of the proposed techniques and compare them to the state
of the art. For the constellation rotation results, it is assumed that $\epsilon_\circ=0.0001$. The adopted channel model is assumed
to be 
\begin{eqnarray}
\mathbf{h}_k\sim\mathcal{CN}(0,\sigma_c^2\mathbf{I}).
\end{eqnarray}
For the sake of comparison, we depict the performance of the optimal block-level beamforming, which can be formulated as \cite{mats}
\begin{eqnarray}\nonumber
\label{SPM_Unicast}
\hspace{-0.5cm}\mathbf{W}(\mathbf{H},\mathbf{\gamma}) = 
& \arg\underset{\mathbf{W}}{\min}& 
\quad \sum_{j=1}^M\|\mathbf{w}_j\|^2 \\
& \text{s.t.}&  \frac{|\mathbf{h}^T_j\mathbf{w}_j|^2}{\sum_{k\neq j,k=1}^M|\mathbf{h}^T_j\mathbf{w}_k|^2+\sigma_z^2 }\geq \gamma_j, j\in K. 
\end{eqnarray}
In Table \ref{Acronym}, we summarize the techniques and the algorithms that are used in Numerical section. 
\begin{table}
	\begin{center}
		\hspace{-0.02cm}\begin{tabular}{|p{3.0cm}|p{1.5cm}|p{1.5cm}|}
			\hline
			Technique&Acronym&Equation Number\\
			\hline
			Optimal beamforming&OB&\eqref{SPM_Unicast},\cite{mats}\\
			\hline
			SLP without rotation&SLP   &\eqref{SLP_conventional_temporal}\\
			\hline
			SLP with rotation&SLPRo&Algorithm 1\\
			\hline
		\end{tabular}
		\vspace{0.2cm}
		\caption{\label{Acronym} Transmit power in dB  vs SNR, 8-QAM }
	\end{center}
\end{table}

In Fig. \ref{Transmit_power_vs_snr_3users}, we depict the performance of the power consumption with respect to target SNR. We compare the conventional optimal beamforming with symbol-level precoding technique with and without constellation rotation (SLPR and SLP). It can be deduced that the optimal beamforming, in general, needs more power to satisfy the required SNR, except for the low SNR region of QPSK modulation, where it needs less power than SLP. In SLP, the precoding strategies aim at exploiting interference, the benefits can be doubtable at low SNR regime as the limiting factor is the noise not the interference. From Fig. \ref{Transmit_power_vs_snr_3users}, it can be deduced that SLPR always outperforms  SLP, the power saving is 1.9 dB for BPSK, and 2.1 dB for QPSK. In comparison with conventional optimal beamforming,  SLPR achieves power saving from 2.5 dB to 4.5 using BPSK modulation, and from 1 dB to 5 dB using QPSK modulation. The constellation rotation in SLP improves the chances of generating constructive interference and therefore better interference exploitation is anticipated.
\begin{figure}[h]
	\includegraphics[scale=0.48]{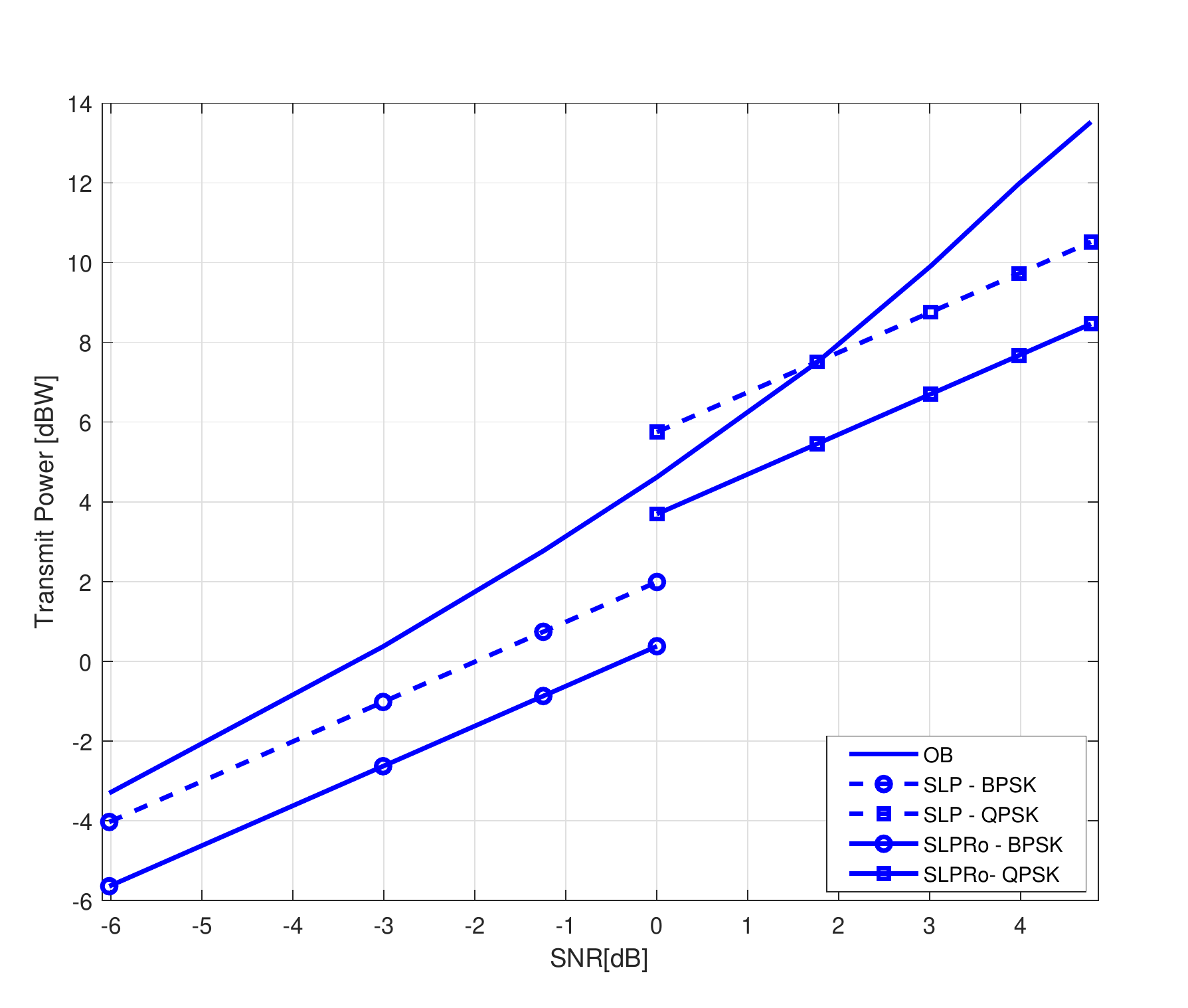}
	\caption{\label{Transmit_power_vs_snr_3users} Transmit power vs. SNR  for different modulations $K=3$, $M=3$. }
\end{figure}

\begin{table}
	\begin{center}
		\hspace{-0.02cm}\begin{tabular}{|p{2.5cm}|p{1.2cm}|p{1.2cm}|p{1.2cm}|}
			\hline
			SNR (dB)&4.7712&7.7815&9.5424\\
			\hline
			Optimal beamforming&11.6878& 15.1144& 17.0172\\
			\hline
			SLP&10.8732   &13.8835  &15.6444\\
			\hline
			SLPRo& 10.6227& 13.6330& 15.3939\\
			\hline
		\end{tabular}
		\vspace{0.2cm}
		\caption{\label{energytable} Transmit power in dB  vs SNR, 8-QAM }
	\end{center}
\end{table}

In Table \ref{energytable}, the transmit power with respect to SNR is illustated for 8-QAM modulation. The required power to achieve SNR target is the highest for the conventional channel beamforming, which matches the result in Fig. \ref{Transmit_power_vs_snr_3users}. SLP techniques outperform the conventional beamforming, and SLP with constellation rotation has the lowest power consumption. However, the amount of power saving achieved by the constellation rotation is 0.25 dB, which is less than the power saving achieved in lower modulation order. For uncorrelated channels, the feasibility of optimizing the constellation rotation lies in the low rate applications.
\subsubsection{Spatially correlated channels} 
The performance of the proposed algorithm is studied for correlated channels. The channel can be modelled as
\begin{eqnarray}
\mathbf{h}_k\sim\mathcal{CN}(0,\mathbf{C}_k), 
\end{eqnarray} 
where $\mathbf{C}_k$ is the spatial correlation. 
\begin{eqnarray}
	[\mathbf{C}_k]_{i,j}=\begin{cases}
	1&,\quad i=j\\
	a^{i-j}&,\quad i> j\\
	{(a^{*})}^{i-j}&,\quad i< j\\
	\end{cases}\\
	\text{where}  a\in \mathbb{C}, |a|<1.
\end{eqnarray}
The impact of spatial correlation on the performance of different precoding techniques in the downlink is depicted in Fig. \ref{spatial correlation} for two users and two antennas scenario. For very highly correlated scenarios, the optimal beamforming in \eqref{SPM_Unicast} outperforms the symbol-level precoding techniques (with and without rotation) in low SNR regime. This means that the cost of exploiting the interference, when the SNR target is 0 dB, is very high. The required power to achieve certain SNR target for SLP with constellation is less than both optimal beamforming and SLP without constellation rotation. The gap between the SLP with and without constellation rotation is 5.2 dB. The rate at which the transmit power increases with SNR target is different from SLP based technique and optimal beamforming; the rate is linear (in dB) for symbol-level precoding schemes and it is non-linear for optimal beamforming.  Moreover, the benefit of having SLP without constellation rotation diminishes when the channel is spatially correlated, thus the joint precoding and constellation rotation is required at these scenario to benefit from the constructive interference.  
\begin{figure}[h]
	\includegraphics[scale=0.48]{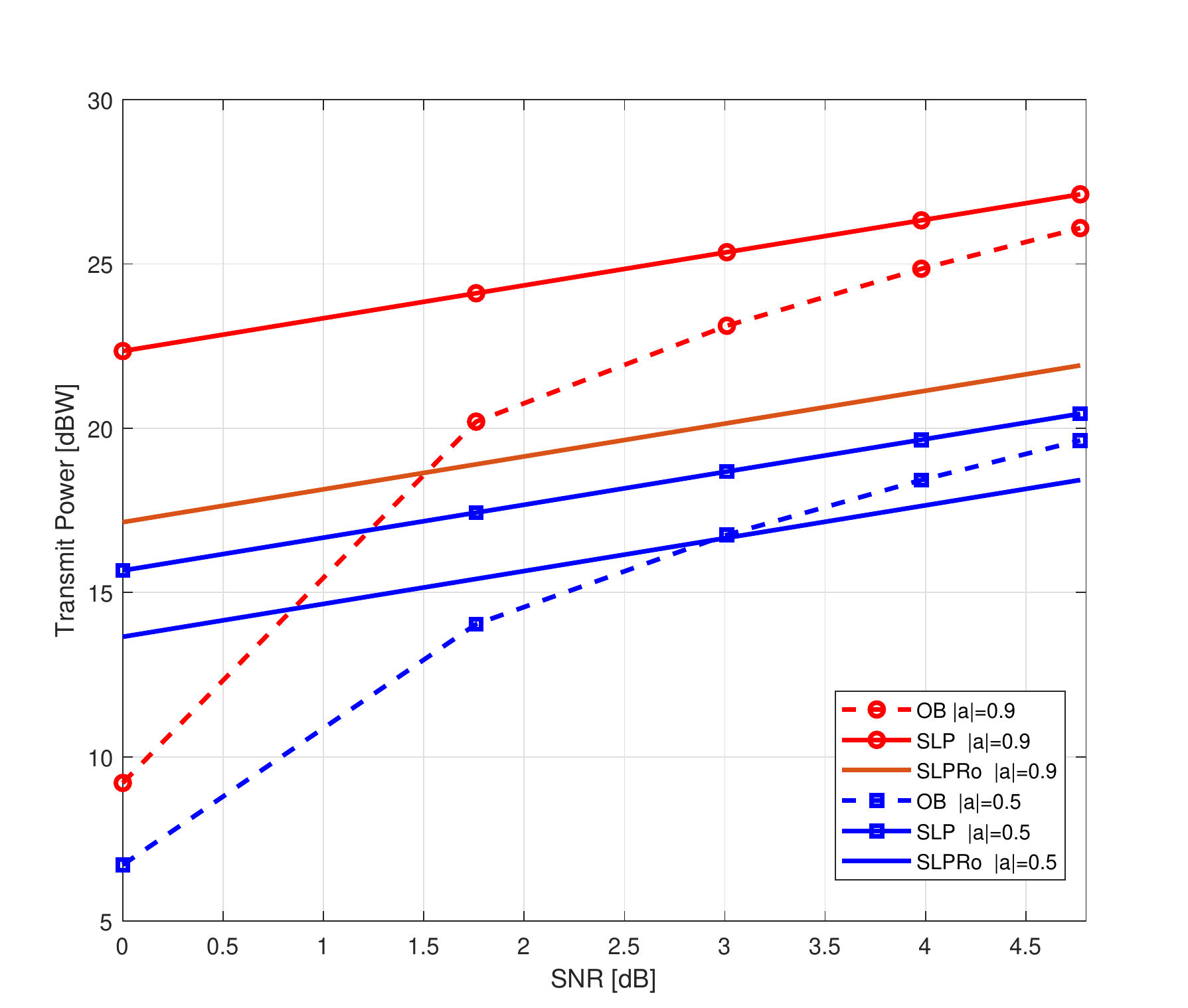}
	\caption{\label{spatial correlation} Transmit power vs. SNR  for different correlation values $K=2$, $M=2$, QPSK. }
\end{figure}

\begin{figure}[h]
	\includegraphics[scale=0.48]{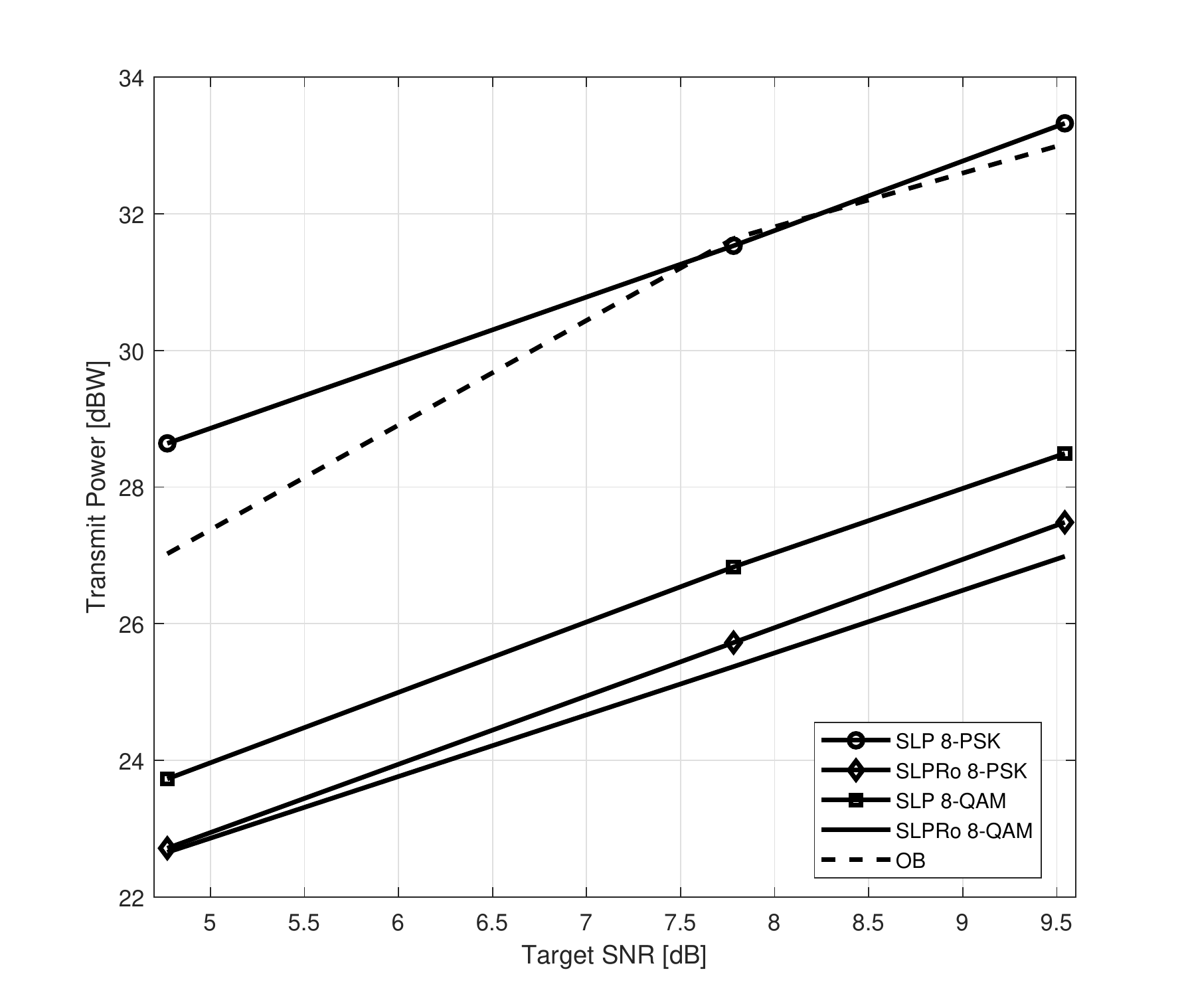}
	\caption{\label{spatial correlation_2} Transmit power vs. SNR   $K=2$, $M=2$ for different modulations. }
\end{figure}
In Fig. \ref{spatial correlation_2}, the performance of SLPR  is studied at high correlation scenario $|\rho|=0.9$ for 8-PSK and 8-QAM modulation. It can be noted that the transmit power to achieve certain SNR is the highest in case of SLP using 8-PSK, even it exceeds the optimal beamforming for almost all studied SNR targets, it can be explained by the required alignment to position each symbol in the correct detection region. This changes when we  optimize the constellation rotation in SLPR, this allows to find the the symbols alignment to exploit the interference and to tackle the spatial correlation. This reduces the required power by about 6 dB, which is considerable in this scenario. In the case of 8-QAM, the transmit power the required to achieve the required SNR is always lower than optimal beamforming by about 5 dB, however, this gap increases when the rotation is optimized  to 6.2 dB. This shows that the constellation can be beneficial for higher order modulation in case of spatial correlation existence. Furthermore, the performance of 8-QAM and 8-PSK with optimized rotation is close to each other in comparison with no rotation where huge performance gap exits. It can be seen that the influence of rotation affects the circular modulation tremendously. 

In Fig. \ref{spatial correlation_3}, the feasibility of employing rotation is studied and compared to conventional precoding and SLP for 16-QAM modulation. For no spatial correlation $a=0$ case, the transmit power saving of SLPR reaches up to 2.6 dB and 0.6 dB in comparison to conventional optimal beamforming and SLP respectively. This shows that the rotation impact decreases with higher modulation order. The impact of rotation becomes more influential for higher spatial correlation. At very high spatial correlation, for example $|a|=0.9$, the transmit power saving of SLPR increases to 5.85 dB and 3.64 dB in comparison to conventional optimal beamforming and SLP. This means that it is still beneficial to employ rotation in high order modulation scenarios if the channel exhibits high spatial correlation. \\
\begin{figure}[h]
\hspace{-0.5cm}	\includegraphics[scale=0.5]{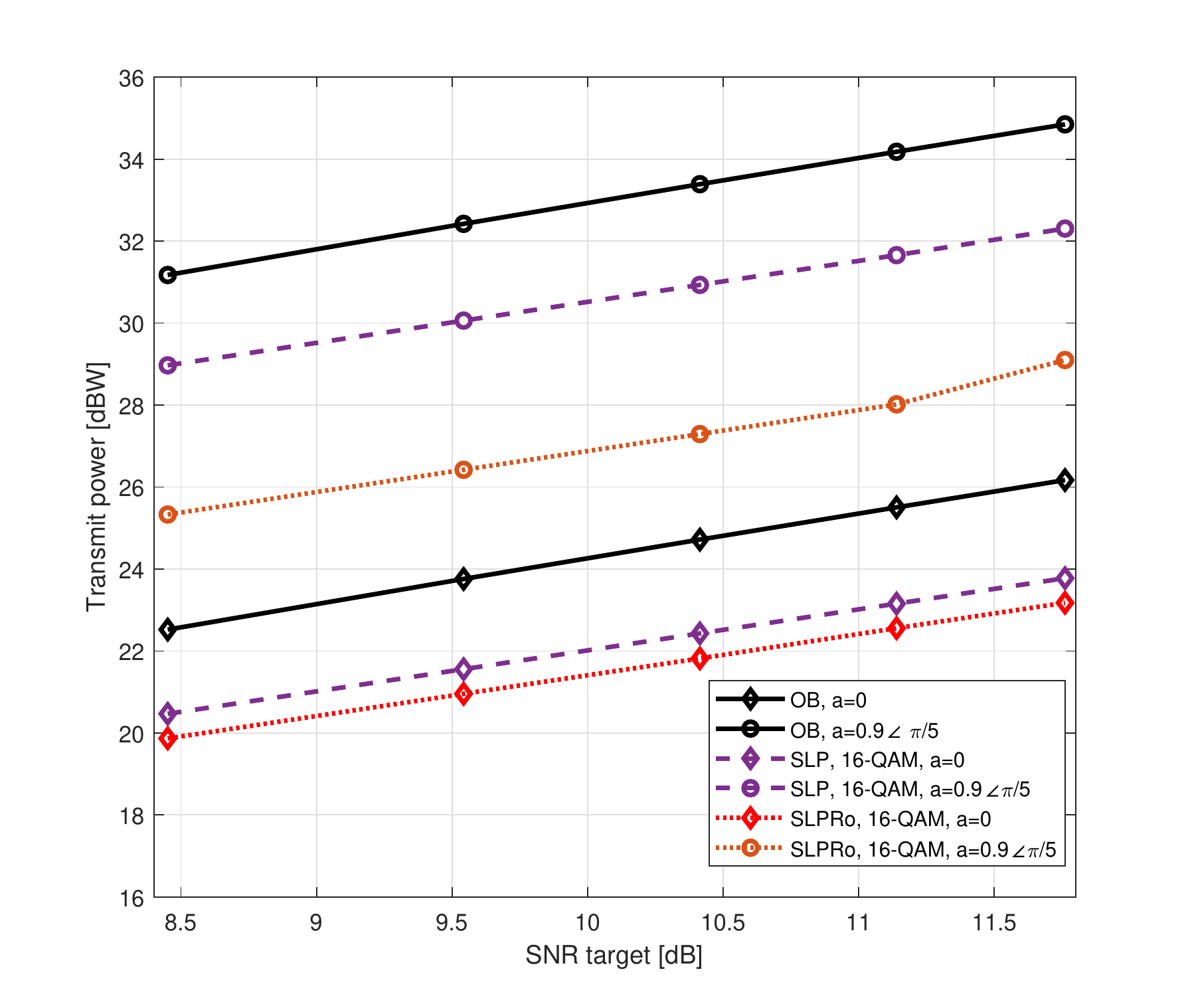}
	\caption{\label{spatial correlation_3} Transmit power vs. SNR   $K=2$, $M=2$ for 16-QAM at different channel correlation values. }
\end{figure}
 To facilitate the importance of constellation, we compare the transmit power of different techniques for a certain channel realization:
\begin{eqnarray}
\label{H_test}
	\mathbf{H}=\begin{bmatrix} -0.4965 + 0.0618i &  0.5403 + 1.0261i\\
	-0.3680 + 0.0010i &  0.2111 + 0.8027i
	\end{bmatrix}.
\end{eqnarray}
To get an indication about this channel, we assess the orthogonality among users by using the following metric
\begin{eqnarray}\nonumber
\text{Proj}(\mathbf{h}_1,\mathbf{h}_2)&=&\frac{\mathbf{h}_1\mathbf{h}^H_2}{\|\mathbf{h}_1\|\|\mathbf{h}^H_2\|}\\\nonumber
&=&	0.9771 - 0.2088i= 0.9992\angle -12.0623.
\end{eqnarray}
The channels in \eqref{H_test} are semi co-linear. This issue represents a bottleneck for conventional block-level MISO precoding techniques \cite{Medra,Medra_transaction,bjorn_patent,bjornson,mats,leakage,leakage_per,Swindlehurst_zeroforcing}. The required transmit power to achieve a certain SNR threshold is shown in Table \ref{energytable1}. It can be noted the power consumption is very high for all precoding and modulation techniques. Moreover, the power consumption for the block-level and SLP has a quite similar performance with slight advantage for the block-level techniques. SLPRo needs less power in comparison with the other technique, and the power saving can reach up to 14 dB. This shows the importance of constellation rotation in co-linear channels scenarios, which relies on finding the symbol alignment that makes the required power is minimum. 

\begin{table}
	\begin{center}
		\hspace{-0.22cm}\begin{tabular}{|p{2.8cm}|p{1.5cm}|p{1.5cm}|p{1.5cm}|}
			\hline
			(SNR [dB], modulation)&(4.771,QPSK)&(4.771,8-QAM)&(4.771,8-PSK)\\
			\hline
			Optimal beamforming&43.112& 43.1118& 43.1118\\
			\hline
			SLP&44.1636 &44.091  &44.8710\\
			\hline
			SLPRo&35.466& 29.6172& 30.6547\\
			\hline
		\end{tabular}
		\vspace{0.2cm}
		\caption{\label{energytable1} Transmit power in dB  vs SNR and modulation order. }
	\end{center}
\end{table}

\subsection{Complexity}
We conclude this section by providing a complexity evaluation of the proposed approaches as a function of algorithm execution time.
Since the proposed optimization problems are tackled resorting to numerical
convex optimization tools \cite{convex_boyd}, analytical expressions for the
complexity are hard to derive. Therefore, the complexity is
numerically evaluated in terms of average running time of
the algorithms over the same machine. Nonetheless, the spatiotemporal precoding introduced herein has an advantage in this
regard, since the optimization procedure applies once per block
and not once per symbol slot.

In Table \ref{runningtable1}, we can compare the run-time to solve the optimization problem for the channl-level and symbol-level precoding. For SLP, we focus on \eqref{SLP_conventional_temporal}, which solves the SLP in single optimization problem considering all data vectors and the reduced set of data vectors from which we can find the rest of the output vectors utilizing symmetry. The channel-level precoding does not get affected by the modulation type. SLP techniques are affected by the modulation type due to the number of possible input data vectors. The algorithm run time is evaluated considering the running for all possible data vectors, which is reflected by the number of constraints in the optimization problem. Moreover, we study the performance for two cases: when all the symbols combinations are considered and when they are halved due to data vectors symmetry as in subsection \ref{complexity_red}. It can be seen that the running time increases with the modulation order. For the conventional SLP, where no constellation rotation is optimized, there is a slight performance difference between 8-PSK and 8-QAM despite the fact they have the same number of symbols combinations. It can be also noted that the constellation rotation increases the required time to solve the optimization problem due to the iterative process. The running time can be reduced by increasing the value of $\epsilon_\circ$, however, this affects the precision of the branch and bound algorithm.
Finally, it can be concluded that exploiting input data vector symmetry reduces the required time to solve the optimization problem due to lesser number of possible input data vectors, constraints, and vector sizes. 
\begin{table}[h]
	\begin{center}
		\hspace{-0.2cm}\begin{tabular}{|p{3.0cm}|p{0.8cm}|p{0.9cm}|p{0.8cm}|p{1.1cm}|}
			\hline
			 Modulation&QPSK&8-QAM&8-PSK&16-QAM\\
			\hline
			Optimal beamforming&0.7258&0.7258&0.7258&0.7258\\
			\hline
			SLP& 0.97& 23.29& 31.7260& 25\\
			\hline
		    SLP/ symmetry& 0.4056 &2.6714  &1.33&3.5281\\
			\hline
			SLPRo &17.9&1940.586&2037.83& N/A\\
			\hline
			SLPRo/ symmetry&7.073&143& 78.84& 43.79 \\
			\hline
		\end{tabular}
		\vspace{0.2cm}
		\caption{\label{runningtable1} Average running time (Sec), $M=2$, $K=2$}
	\end{center}
\end{table}
\section{Conclusions}

In this work, we have proposed a complexity reduction scheme that utilizes the constellations symmetry, which results in reducing the number of constraints or the number of optimization to ne solved to $1/4$ of the original SLP problem. Another contribution is an algorithm that jointly optimizes the transmit precoding and constellation
rotation in the downlink of multiuser MISO system. The motivation is to improve SLP capabilities by optimally changing the symbols alignment  to increase the probability of creating constructive interference. Therefore, a better interference exploitation is expected; a better energy efficiency of the  multiuser MISO systems.  The power minimization has been formulated into a non-convex problem containing bilinear and constant modulus constraints. To find the optimal rotation for each user and the output vector for each input data vector,  an algorithm based on semidefinite relaxation and the branch and bound strategy has been proposed. This algorithm reformulates the non-convex constraint into set of affine constraints to treat the problem efficiently.  
 
We have studied the performance of the proposed algorithm with respect to modulation order and channel correlation. It has been shown that the energy efficiency improves tremendously in the correlated channels with the constellation rotation,  where the power saving can reach up to  6 dB in highly correlated channels in comparison with the uncorrelated channel where it can reach up to 2 dB. Moreover, the impact of constellation rotation on higher modulations appears to decline, especially in uncorrelated channel scenarios. This effect still holds for correlated channels, however, it is less pronounced. Finally, the proposed algorithm has a higher complexity than conventional channel-level and symbol-level techniques, this requires to further investigate to reduce the complexity without jeopardizing the achieved power savings. This paper presents a generalization step to optimize the symbol-level precoding in different contexts such as mitigating the channel non-linearities, reducing the peak to average power ratio, etc; to see the impact of the constellation rotation on the performance for different scenarios.

\end{document}